\documentclass{elsarticle}

\usepackage{fullpage}

\usepackage[ruled]{algorithm2e}
\usepackage{graphicx}
\usepackage{amsmath}
\usepackage{amssymb}
\usepackage[american]{babel}
\usepackage{dsfont}
\usepackage{todonotes}
\usepackage{xspace}
\usepackage{url}
\usepackage[pdfborder={0 0 0}]{hyperref} 
\usepackage{subfig}
\usepackage{overpic}
\usepackage{booktabs}

\graphicspath{{./pictures/pdf/},{./pictures/ps/},{./pictures/png/},{./pictures/jpg/}}

\newtheorem{theorem}{Theorem}
\newtheorem{lemma}[theorem]{Lemma}
\newtheorem{corollary}[theorem]{Corollary}

\def\QED{\ensuremath{{\Box}}}
\def\markatright#1{\leavevmode\unskip\nobreak\quad\hspace*{\fill}{#1}}
\newenvironment{proof}
 {\begin{trivlist}\item[\hskip\labelsep{\bf Proof.}]}
 {\markatright{\QED}\end{trivlist}}

\vfuzz \hfuzz

\SetAlFnt{\small}
\SetAlCapFnt{\small}
\SetAlCapNameFnt{\small}
\SetAlCapHSkip{0pt}
\IncMargin{-\parindent}

\newcommand{\old}[1]{{}}

\begin{document}

\begin{frontmatter}


\title{Particle Computation: Complexity, Algorithms, and Logic\footnote{
This paper provides full details for and combines results of a number of different extended abstracts that have appeared in the
International Symposium on Algorithms and Experiments for Sensor Systems, Wireless Networks and Distributed Robotics (ALGOSENSORS 2013)~\cite{Becker2013f},
IEEE International Conference on Robotics and Automation (ICRA 2014)~\cite{Becker2014} and
ICRA 2015~\cite{shad2015particle}.
See video from the  31st International Symposium on Computational Geometry (SoCG'15)~\cite{bmd+-pcdfbm-15}
for illustrations and animations.}
}

\author[uh]{Aaron T.~Becker}
\ead{atbecker@uh.edu}
\author[mit]{Erik D.~Demaine}
\ead{edemaine@mit.edu}
\author[tubs]{S{\'a}ndor P.\ Fekete\corref{corr}}
\ead{s.fekete@tu-bs.de}
\author[uh]{Jarrett Lonsford}
\ead{JLLonsford@uh.edu}
\author[brandeis]{Rose Morris-Wright}
\ead{rmorriswright@brandeis.edu}

\address[uh]{Department of Electrical and Computer Engineering, University of Houston, TX, USA,}
\address[mit]{Computer Science and Artificial Intelligence Laboratory, MIT, Cambridge, MA, USA,}
\address[tubs]{Department of Computer Science, TU Braunschweig, Germany,}
\address[brandeis]{Mathematics, Brandeis University, Waltham, MA, USA}



\begin{abstract}
We investigate algorithmic control of a large swarm of mobile particles (such as robots, sensors, or building material) that move in a 2D workspace using a global input signal (such as gravity or a magnetic field).
Upon activation of the field, each particle moves maximally in the same direction until forward progress is blocked by a stationary obstacle or another stationary particle.
In an open workspace, this system model is
of limited use because it has only two controllable degrees of freedom---all
particles receive the same inputs and move uniformly. We show that adding a
maze of obstacles to the environment can make the system drastically more complex but also more useful.

We provide a wide range of results for a wide range of questions. These can be subdivided into {\em external} algorithmic 
problems, in which particle configurations serve as input for computations that are performed 
elsewhere, and {\em internal} logic problems, in which the particle configurations themselves are used
for carrying out computations.

For external algorithms, we give both negative and positive results.
If we are {\em given} a set of stationary obstacles, we prove that it is NP-hard to decide whether a given initial configuration 
of unit-sized particles can be transformed into a desired target configuration. Moreover, we show that finding a control sequence of minimum length is PSPACE-complete.
We also work on the inverse problem, providing constructive algorithms to {\em design} workspaces that efficiently implement arbitrary permutations between different configurations.

For internal logic, we investigate how arbitrary computations can be implemented.
We demonstrate how to encode \emph{dual-rail logic} to
build a universal logic gate that concurrently evaluates {\sc and, nand, nor,} and {\sc or} 
operations. Using many of these gates and appropriate interconnects, we can
evaluate any logical expression. 
However, we establish that
simulating the full range of complex interactions present in arbitrary digital circuits
encounters a fundamental difficulty: a {\sc fan-out} gate cannot be generated.
We resolve this missing component with the help of 
2$\times$1 particles, which can create {\sc fan-out} gates that produce multiple copies
of the inputs.  
Using these gates we provide rules for replicating arbitrary digital circuits.


\end{abstract}


\end{frontmatter}

\section{Introduction}\label{sec:intro}

\begin{center}
\begin{minipage}{0.8\textwidth} 
{\em ``Programmable matter refers to a substance that has the ability to change its
physical properties (shape, density, moduli, conductivity, optical properties,
etc.) in a programmable fashion, based upon user input or autonomous sensing.
The potential applications are endless, e.g., smart materials, autonomous
monitoring and repair, or minimal invasive surgery. Thus, there is a high
relevance of this topic to industry and society in general, and much research
has been invested in the past decade to fabricate programmable matter. However,
fabrication is only part of the story: without a proper understanding of how to
program that matter, complex tasks such as minimal invasive surgery will be out
of reach.''} \cite{dag16}
\end{minipage}
\end{center}

Since the first visions of
massive sensor swarms, more than ten years of work on sensor networks have
yielded considerable progress with respect to hardware miniaturization.
The original visions of ``Smart Paint''~\cite{AAC00}
or ``Smart Dust''~\cite{kahn00emerging} have triggered a considerable amount
of theoretical research on swarms of {\em stationary}
processors, e.g., the work in \cite{fk-gbrls-06,fk-trsn-07,fkp+-nbtrs-04,kfp+-dbrte-06}.
Recent developments in the ability to design, produce, and control
particles at the micro and nanoscale and the rise of 
possible applications, e.g., targeted drug delivery, micro and nanoscale construction, and Lab-on-a-Chip,
motivate the study of large swarms of {\em mobile} objects.
But how can we control such a swarm with only limited computational
power and a lack of individual control by a central authority?
Local, robotics-style motion control by the particles themselves appears hopeless because 
the capacity for storing energy for computation, communication, and motion control is proportional to the volume, but volume shrinks
with the third power of particle length.

A possible answer lies in applying a global, external force to all particles in the swarm. 
This resembles the logic puzzle Tilt \cite{Tilt}, slide and merge games such as the 2048 puzzle \cite{abdelkader20162048}, and dexterity ball-in-a-maze puzzles
such as Pigs in Clover and Labyrinth, which involve tilting a board to cause all
mobile pieces to roll or slide in a desired direction.
Problems of this type are also similar to sliding-block puzzles with fixed obstacles
\cite{Demaine2000,Hoffmann2000,Holzer2004,HEARN200572in2005},
except that all particles receive the same control inputs, as in the Tilt puzzle.
In the real world,
driving ferromagnetic particles with a magnetic resonance imaging (MRI) scanner gives a milli-scale example of this challenge~\cite{Vartholomeos2012}. 
At the micro-scale, Becker et al.~\cite{Becker2013a} demonstrate how to apply
a magnetic field to simultaneously move cells containing iron particles in a specific
direction within a fabricated workspace (see Fig.~\ref{fig:magneticPyriformis}). 
Other recent examples include
using the global magnetic field from an MRI to guide magneto-tactic bacteria
through a vascular network to deliver payloads at specific locations
\cite{Chanu2008} and using electromagnets to steer a
magneto-tactic bacterium through a micro-fabricated maze \cite{Khalil2013b}; however, this still involves
only individual particles at a time, not the parallel motion of a whole, massive swarm.
How can we manipulate the overall swarm with coarse global control,
such that individual particles arrive at multiple different destinations in a (known) complex vascular
network such as the one in Fig.~\ref{fig:vascularnet}?
And how can we use the complex interaction of the particles to carry out complex computations
{\em within} the swarm?

\begin{figure}
\centering
\subfloat[][\label{fig:magneticPyriformis}(Left) After feeding iron particles
to ciliate eukaryon (\emph{Tetrahymena pyriformis}) and magnetizing the
particles with a permanent magnet, the cells can be turned by changing the
orientation of an external magnetic field (see colored paths in the center
image). (Right) Using two orthogonal Helmholz electromagnets, Becker et al.~\cite{Becker2013a}
demonstrated steering many living magnetized \emph{T. pyriformis} cells.
All cells are steered by the same global field. ]
{\begin{overpic}[width=\columnwidth]{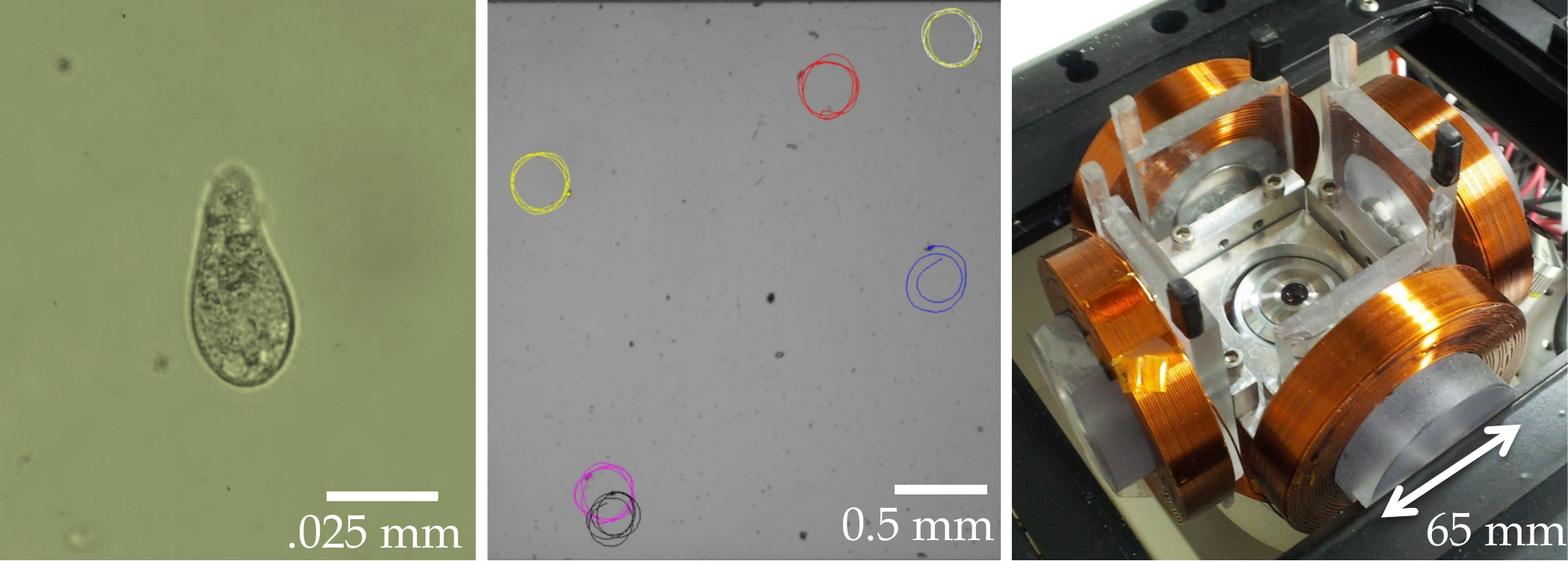}
\end{overpic}
}
\hspace{.5em}
\subfloat[][
\label{fig:vascularnet}
Biological vascular network (cottonwood leaf). (Photo: \href{http://www.tssphoto.com/index.php?p=980}{Royce Bair/Flickr/Getty Images}.)  Given such a network along with initial and goal positions of $N$ particles,  is it possible to bring each particle to its goal position using a global control signal?  Note that this arrangement is \emph{not} a tree, but a graph structure with many cycles. {\sc Matlab} code for driving $N$ particles through this network is available at \href{http://www.mathworks.com/matlabcentral/fileexchange/42892}{http://www.mathworks.com/matlabcentral/fileexchange/42892}.

\label{fig:vascularNetwork}]
{\begin{overpic}[width=0.7\columnwidth, angle = 270]{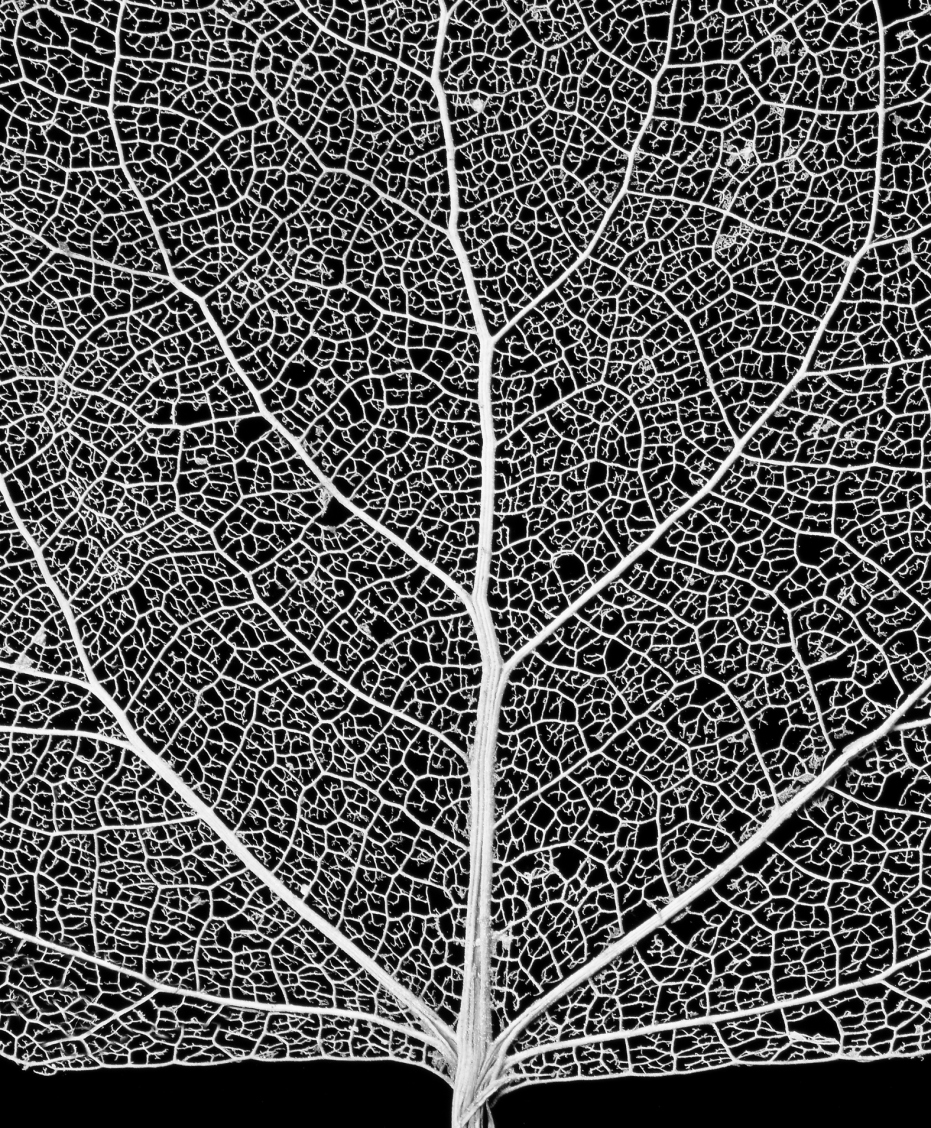}\end{overpic}
}
\caption{(Top) State of the art in controlling small objects by force fields. (Bottom) A complex vascular network,
forming a typical environment for the parallel navigation of small objects. This paper investigates parallel navigation in discretized 2D environments.
}
\end{figure}

All this gives rise to the following two families of problems, which we denote by {\em External Computation} and {\em Internal Computation}. 

\bigskip
{\bf External Computation.}
Considering the particle swarm as input for a given algorithmic problem, we are faced with a number
of questions that need to be resolved externally, such as the following.

\begin{enumerate}
\item 
Given a map of an environment, such as the vascular network shown
in Fig.~\ref{fig:vascularNetwork}, along with initial and goal positions for each particle, 
does there exist a sequence of inputs that will bring each particle to its goal position?
\item 
Given a map of an environment, such as the vascular network shown
in Fig.~\ref{fig:vascularNetwork}, along with initial and goal positions for each particle,
what is the shortest sequence of moves that will bring each particle to its goal position?
\item
Given initial and goal positions for each particle in a swarm, how can we design a set of obstacles
and a sequence of moves, such that each particle reaches its goal position?
\end{enumerate}

Deliberate use of existing stationary obstacles leads to a wide range of possible particle configurations.
In the first part of the paper (Section~\ref{sec:mazes} and  Section~\ref{sec:matrices}), we give answers to the first two questions by showing that they 
may lead to computationally difficult situations.
We also develop several positive results for the third question (again in Section~\ref{sec:matrices}).
The underlying idea is to construct artificial
obstacles (such as walls) that allow arbitrary rearrangements of a given two-dimensional particle swarm. 
For clearer notation, we will formulate the relevant statements in the language of matrix operations,
which is easily translated into plain geometric language.
This paper investigates these problems in 2D discretized environments, leaving 3D and continuous environments for future work.

\bigskip
{\bf Internal Computation.}
Considering the particle swarm as a complex system that can be reconfigured
in various ways, we are faced with issues of the computational power of the swarm itself,
such as the following.

\begin{enumerate}
\item Can the complexity of particle interaction be exploited to model logical operations?
\item Are there limits to the computational power of the particle swarm?
\item How can we achieve computational universality with particle computation?
\end{enumerate}

In the second part of the paper (Section~\ref{sec:logic}), we give precise answers to all of these
questions. In particular, we show that the logical operations 
{\sc and, nand, nor}, and {\sc or} can be implemented in our model using dual-rail logic. Using terminology from electrical engineering, we call these components that calculate logical operations \emph{gates}.
We establish a fundamental limitation for particle interactions: we cannot duplicate the output of a chain of gates without also duplicating the chain of gates. 
This means a {\sc fan-out} gate cannot be generated.
We resolve this missing component with the help of
2$\times$1 particles, which can be used to create {\sc fan-out} gates that produce multiple copies
of the inputs without needing duplicate gates.  
Using  {\sc fan-out} gates, we provide rules for replicating arbitrary digital circuits.

In the following, we start by a brief formal problem definition (Section~\ref{sec:prelim}) and a discussion of
related work (Section~\ref{sec:related}), then continue to provide the results on External Computation
(Sections~\ref{sec:mazes} and \ref{sec:matrices}) and Internal Computation (Sections~\ref{sec:logic} and \ref{sec:Design}).
We conclude with future work (Section~\ref{sec:conclusion}).

\section{Problem Definition}
\label{sec:prelim}

\subsection{Model}
  
Our model is based on 
the following rules: 
\begin{enumerate}
\item A planar  grid \emph{workspace} $W$ contains a number of unit-size particles and some fixed unit-square obstacles.   A grid cell is referenced by its Cartesian coordinates $\mathbf{x}=(x,y)$, and is either \emph{free} for possible occupation by a particle, or a permanent \emph{obstacle}, which may never be occupied by a particle.
Each particle occupies one grid cell.
\item All particles are commanded in unison: the valid commands are  ``Go Up" ($u$), ``Go Right" ($r$), ``Go Down" ($d$), and ``Go Left" ($\ell$).  
\item The particles all move in the commanded direction  until forward progress is blocked by a stationary obstacle or another blocked particle.
 A \emph{command sequence} $\mathbf{m}$ consists of an ordered sequence of moves $m_k$, where each $m_k\in\{u,d,r,\ell\}$.  A representative command sequence is $\langle u,r,d,\ell,d,r,u,\ldots\rangle$. We assume that $W$ is bounded by obstacles and issue each command long enough for the particles to move to their maximum extent.
\end{enumerate}

The algorithmic decision problem {\sc GlobalControl-ManyParticles}  is to decide whether a given puzzle is solvable. In other words, given a fixed workspace and a start and goal location for each particle, the algorithm determines the existence of a sequence of moves that move the particles to their goal locations. See Fig.~\ref{fig:noSolutionASolution}
for two simple instances.

\begin{figure}
\centering
\begin{overpic}[width=0.65\columnwidth]{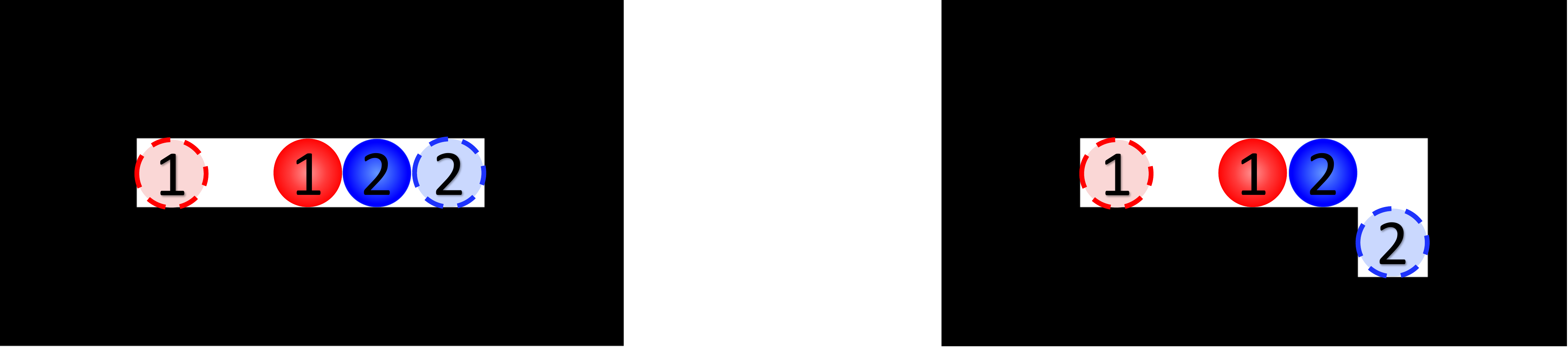}
\end{overpic}
\caption{
\label{fig:noSolutionASolution}
In this image, black cells are fixed, white cells are free, solid discs are individual particles, and goal positions are dashed circles.  For the simple instance on the left, it is impossible to maneuver both particles to end at their goals. The instance on the right has a finite solution: $\langle r,d,\ell \rangle$.  
}
\end{figure}

\section{Related Work}\label{sec:related}

Related work is categorized into underactuated control, manipulation, and computation. 

\subsection{Underactuated Control}

\paragraph{Large Robot Populations}
Due to the efforts of roboticists,  biologists,  and chem\-ists  (e.g., \cite{Rubenstein2012,Ou2013,Chiang2011}),
it is now possible to make and field large ($N=10^3$--$10^{14}$) populations of simple robots.  Potential applications for these robots include targeted medical therapy, sensing, and actuation. With large populations come two fundamental challenges:  how to (1)  perform state estimation for the robots and (2) control these robots.

Traditional approaches often assume independent control signals for each robot, but each additional independent signal requires bandwidth and engineering. These bandwidth requirements grow at $O(N)$.
Using independent signals becomes more challenging as the robot size decreases. 
  Especially at the micro- and nano-scales, it is not practical to encode autonomy in the robots.  
  Instead, the robots are controlled and interacted with using global control signals. For this reason, it may be more appropriate to call the moving agents \emph{particles} and label the external control system as the robot.

More recently, robots have been constructed with physical heterogeneity so that they respond differently to a global, broadcast control signal.  Examples include \emph{scratch-drive microrobots}, actuated and controlled by a DC voltage signal from a substrate \cite{Donald2013};   magnetic structures  with different cross-sections that could be independently steered \cite{Floyd2011};   \emph{MagMite} microrobots with different resonant frequencies and a global magnetic field \cite{Frutiger2008}; and  magnetically controlled nanoscale helical screws constructed to stop movement at different cutoff frequencies of a global magnetic field
\cite{Peyer2013}. In previous work with robots modeled as nonholonomic unicycles, we showed that an inhomogeneity in turning speed is enough to make even an infinite number of robots controllable with regard to position. All these approaches show promise, but they require precise state estimation and heterogeneous robots.
At the molecular scale, there is a bounded number of modifications that can be made to differentiate robots. 
 In addition, the control law computation requires at best a summation over all the robot states $O(N)$ \cite{Becker2012k,Becker01112014} and at worst a matrix inversion $O(N^{2.373})$\cite{Becker2012}. 

In this paper we take a different approach.  We assume a population of approximately identical planar particles (which could be small robots) and  one global control signal that contains the direction all particles should move.  In an open environment, this system is not controllable because the particles move uniformly---implementing any control signal translates the entire group identically;  however, an obstacle-filled workspace allows us to break this symmetry. In previous practical work~\cite{Becker2013b}, we showed that if we can command the particles to move one unit distance at a time, some goal configurations have easy solutions. Given a large free space, we have an algorithm showing that a single obstacle is sufficient for position control of $N$ particles (video of position control: \url{http://youtu.be/5p_XIad5-Cw}).  This result required incremental position control of the group of particles, i.e. the ability to advance them a uniform, fixed distance.  This is a strong assumption and one that we relax in this work. 

\paragraph{Dexterity Games}
The problem we investigate is strongly related to dexterity puzzles---games that typically involve a maze and several balls that should be maneuvered to goal positions. Such games have a long history. \emph{Pigs in Clover}, involving steering four balls through three concentric incomplete circles, was invented in 1880 by Charles Martin Crandall. 
Dexterity games are dynamic and depend on the manual skill of the player.  Our problem formulation also applies the same input to every agent but imposes only kinematic restrictions on agents.  This is most similar to the gravity-fed logic maze \href{http://www.thinkfun.com/tilt}{\emph{Tilt}\texttrademark},
invented by Vesa Timonen and Timo Jokitalo and distributed by \href{http://www.thinkfun.com}{ThinkFun}
since 2010 \cite{Tilt}.

 \paragraph{Sliding-Block Puzzles}
Sliding-block puzzles use rectangular tiles that are constrained to move in a 2D workspace. The objective is to move one or more tiles to desired locations. They have a long history.
Hearn~\cite{hearn2005complexity} and Demaine et al.~\cite{Demaine2009} showed that tiles can be arranged to create logic gates and used this technique to prove {\sc pspace}-completeness for a variety of sliding-block puzzles.  Hearn expressed the idea of building computers from the sliding blocks---many of the logic gates could be connected together, and the user could propagate a signal from one gate to the next by sliding intermediate tiles.  This requires the user to know precisely which sequence of gates to enable and disable.  In contrast to such a hands-on approach, with our architecture we can build circuits, store parameters in memory, and then actuate the entire system in parallel using a global control signal.

 \subsection{Manipulation}
 
\paragraph{Computational Geometry: Robot Box-Pushing}
Many variations of block-pushing puzzles have been explored from a computational complexity viewpoint with a seminal paper proving NP-hardness by Gordon Wilfong in 1991 \cite{Wilfong1991}.
 The general case of motion-planning when each command moves particles a single unit in a world composed of even a single robot and both \emph{fixed} and \emph{moveable} squares is in the complexity class {\sc pspace}-complete
\cite{Demaine2009,Dor1999,Hoffmann2000}.

Ricochet Robots \cite{Engels2005}, Atomix \cite{Holzer2004}, and PushPush \cite{Demaine2000} have the same constraint that particles must move to their full extent, once they have been set in motion. This constraint reflects physical realities where, due to uncertainties in sensing, control application, and dynamic models, precise quantified movements in a specified direction are not possible.
Instead, the input can be applied for a long period of time, guaranteeing that all particles move to their fullest extent. 
In these games the particles move to their full extent with each input, but each particle can be actuated individually.  
The problem complexity with global inputs to all particles is addressed in this paper.


\paragraph{Sensorless Manipulation}
The algorithms in the second half of our paper do not require feedback, and we have drawn inspiration from work on sensorless manipulation~\cite{Erdmann1988}.
Sensorless manipulation explicitly maintains the set of all possible
part configurations and selects a sequence of actions to reduce the size of this set and drive it toward some goal configuration.
Carefully selected primitive operations can make this easier.
Sensorless manipulation strategies often use a sequential composition of
primitive operations, ``squeezing'' a part either virtually with a programmable force field or simply between two flat, parallel plates~\cite{Goldberg1993}.
Some sensorless manipulation strategies take advantage of limit cycle behavior,
for example, engineering fixed points and basins of attraction so that parts only exit a feeder when they reach the correct orientation~\cite{Lynch2002,Murphey2005}.
These two strategies have been applied to a much wider array of mechanisms such as vibratory bowls and tables~\cite{Goemans2006,Vose01082009,vose2012sliding} or assembly lines~\cite{Akella2000,Goldberg1993,Stappen2002},
and have also been extended to situations with stochastic uncertainty~\cite{Goldberg1999,Moll2002} and closed-loop feedback~\cite{Akella1999,Murphey2004a}.

 \subsection{Computation}
\paragraph{Parallel Algorithms: SIMD}

Another related area of research is Single Instruction Multiple Data (SIMD)
parallel algorithms \cite{Leighton1991}.  In this model, multiple processors
are all fed the same instructions to execute, but they do so on different
data.  This model has some flexibility, for example, allowing command execution selectively only on
certain processors and no operations (NOPs) on the remaining processors. 

Our model is actually more extreme. The particles all respond in effectively
the same way to the same instruction.  The only difference is their location
and which obstacles or particles will block them.  In some sense,
our model is essentially Single Instruction, Single Data, Multiple Locations.

Our efforts have similarities with \emph{mechanical computers},  computers
constructed from mechanical, not electrical components. For a fascinating
nontechnical review, see \cite{McCourtney1999}.  These devices have a rich
history, from the \emph{Pascaline}, an adding machine invented in 1642 by a
nineteen-year old Blaise Pascal, to Herman Hollerith's punch-card tabulator in
1890, to the mechanical devices of IBM culminating in the 1940s.  These devices
used precision gears, pulleys, or electric motors to carry out calculations.
Though our implementations in this paper are rather basic, 
we require none of these precision elements to achieve computational universality---merely unit-size obstacles and
sliding particles sized 2$\times$1 and 1$\times$1.

\paragraph{Collision-Based Computing}
Collision-based computing has been defined as \emph{``computation in a structureless medium populated with mobile objects.''}  For a survey of this area, see the excellent collection~\cite{Adamatzky2012}. Early examples include the billiard-ball computer proposed by Fredkin and Toffoli using only spherical balls and a frictionless environment composed of elastic collisions with other balls and with angled walls \cite{Fredkin1982ConservativeLogic}. Another popular example is Conway's {\em Game of Life}, a cellular automaton governed by four simple rules~\cite{berlekamp2001winning}. Cells live or die based on the number of neighbors. These rules have been examined in depth and used to design a Turing-complete computer \cite{Adamatzky2002,rendell2011universal}.  Game of life scenarios and billiard-ball computers are fascinating but lack a physical implementation.  In this paper we present a collision-based system for computation and provide a physical implementation.

\paragraph{Programmable Matter}
Clearly there is a wide range of interesting scenarios for developing approaches to programmable matter.
One such model is the \emph{abstract Tile-Assembly Model} (aTAM) by Winfree~\cite{Winf98,WLWS98,LaWiRe99}, which has 
sparked a wide range of theoretical and practical research. In this model, unit-sized tiles
interact and bond with the help of differently labeled edges, eventually composing complex assemblies.
Even though the operations and final objectives in this model are quite different from our particle computation with global
inputs (e.g., key features of the aTAM are that tiles can have a wide range of different edge types, and
that they keep sticking together after bonding), there is
a remarkable geometric parallelism to a key result of our present paper:
while it is widely believed that at the most basic level of interaction (called {\em temperature 1}),
computational universality {\em cannot} be achieved~\cite{LSAT1,ManuchTemp1,IUNeedsCoop} in the aTAM with only unit-sized tiles, 
 recent work~\cite{fhp+-ucapt-15} shows that computational universality {\em can} be achieved as soon as even slightly bigger tiles are used. 
This resembles the results of our paper, which shows that unit-size particles are insufficient for universal computation, but employing bigger particles suffices.

\section{Mazes}\label{sec:mazes} 
We prove that the general problem defined in Section~\ref{sec:prelim} is computationally intractable.

\begin{theorem}
  {\sc GlobalControl-ManyParticles} is NP-hard:
  given a specified goal location and an initial configuration of movable particles and fixed obstacles, it is
  NP-hard to decide  
  if a move sequence exists that ends with some particle at the goal location.
\end{theorem}

\begin{proof}
  We prove hardness by a reduction from 3SAT.
  Suppose we are given $n$ Boolean variables $x_1, x_2, \dots, x_n$, and
  $m$ disjunctive clauses $C_j = U_j \vee V_j \vee W_j$,
  where each literal $U_j,V_j,W_j$ is of the form $x_i$ or $\neg x_i$.
  We construct an instance of {\sc GlobalControl-ManyParticles} that has a
  solution if and only if all clauses can be satisfied by a truth
  assignment to the variables. This instance is composed of \emph{variable gadgets} for setting individual variables  {\sc True} or  {\sc False}, \emph{clause gadgets} that construct the logical  {\sc or} of groupings of three variables, and  a \emph{check gadget} that constructs the logical  {\sc and} of all the clauses.  A particle is only delivered to the goal location if the variables have been set in such a way that the formula evaluates to  {\sc True}.

\paragraph{Variable gadgets}
  For each variable $x_i$ that appears in $k_i$ literals,
  we construct $k_i$ instances of the \emph{variable gadget} $i$
  shown in Fig.~\ref{fig:VariableGadget},
  with a particle initially at the top of the gadget.
  The gadget consists of a tower of $n$ levels,
  designed for the overall construction to make $n$ total variable choices.
  These choices are intended to be made by a move sequence of the form
 $\langle d$, $l/r$, $d$, $l/r$, \dots, $d$, $l/r$, $d$, $r\rangle$,
  where the $i$th $l/r$ choice corresponds to setting variable $x_i$
  to either {\sc True} ($l$) or {\sc False} ($r$).
  Thus variable gadget $i$ ignores all but the $i$th choice by
  making all other levels lead to the same destination via both $l$ and~$r$.
  The $i$th level branches into two destinations chosen by either $l$ or $r$,
  which correspond to $x_i$ being set {\sc True} or {\sc False}, respectively.

  In fact, the command sequence may include multiple $l$ and $r$ commands
  in a row, in which case the last $l/r$ before a vertical $u/d$ command
  specifies the final decision made at that level, and the others can be
  ignored.
  The command sequence may also include a $u$ command, which undoes a $d$
  command if done immediately after or else does nothing; thus we can simply
  ignore the $u$ command and the immediately preceding $d$, if it exists.
  We can also ignore duplicate commands (e.g., $d, d$ becomes $d$) and
  remove any initial $l/r$ command.
  After ignoring these superfluous commands, assuming a particle reaches one
  of the output channels, we obtain a sequence in the canonical form
  $\langle d$, $l/r$, $d$, $l/r$, \dots, $d$, $r \rangle$ as desired,
  corresponding uniquely to a truth assignment to the $n$ variables.
  (If no particle reaches the output port, it is as if the variable is neither
  {\sc True} nor {\sc False}, satisfying no clauses.)
  Note that all particles arrive at their output ports at exactly the same time.

\begin{figure}[htbp]
\centering
\subfloat[][\label{fig:VariableGadget1} 
 Variable $x_1$ set {\sc True}  for $i=1$]
{\begin{overpic}[height=0.16\columnwidth]{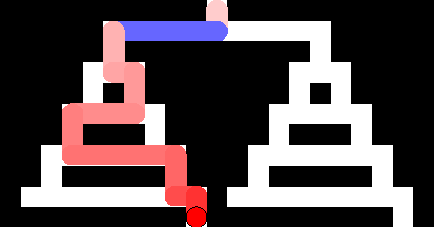}
\end{overpic}
}
\hspace{.1em}
\subfloat[][\label{fig:VariableGadget2} 
  $x_2$ set {\sc False}, $i=2$]
{\begin{overpic}[height=0.16\columnwidth]{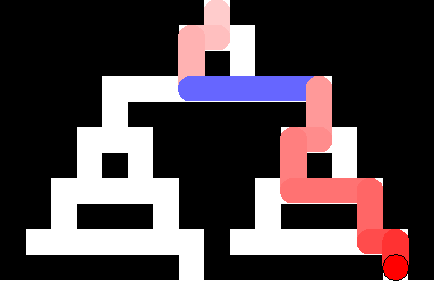}
\end{overpic}
}
\hspace{.1em}
\subfloat[][\label{fig:VariableGadget3} 
   $x_3$ set {\sc True}, $i=3$]
{\begin{overpic}[height=0.16\columnwidth]{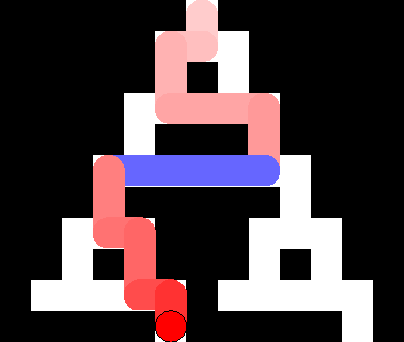}
\end{overpic}
}
\hspace{.1em}
\subfloat[][\label{fig:VariableGadget4} 
   $x_4$ set {\sc False}, $i=4$]
{\begin{overpic}[height=0.16\columnwidth]{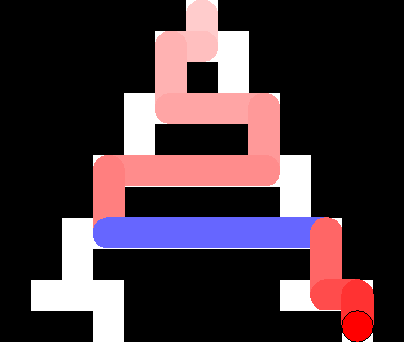}
\end{overpic}
}
\caption{
\label{fig:VariableGadget}
  Variable gadgets that are assigned a truth value by executing a sequence of $\langle d, \ell/r \rangle$ moves. The $i$th $\ell/r$ choice sets the variable $x_i$ to {\sc True} or {\sc False} by putting the particle in the left or right column. This selection move is shown in blue. Each gadget is designed to respond to the $i$th choice but ignore all others. This lets us make several copies of the same variable by making multiple gadgets with the same $i$. In the figure $n=4$, and the input sequence $\langle d,\ell,d,r,d,\ell,d,r,d,r,d\rangle$ sets ({\sc $x_1$=True, $x_2$=False, $x_3$=True, $x_4$=False}).}
\end{figure}

\paragraph{Clause gadgets}
  For each clause, we use the {\sc or} gadget shown in Fig.~\ref{fig:3InputOR}.
  The {\sc or} gadget has three inputs corresponding to the three literals, and input particles are initially at the top of these inputs.
  For each positive literal of the form $x_i$, we connect the corresponding
  input to the left output of an unused instance of variable
  gadget~$i$.
  For each negative literal of the form $\neg x_i$, we connect the
  corresponding input to the right output of an unused instance of a
  variable gadget~$i$.
  (In this way, each variable gadget gets used exactly once.)

  We connect the variable gadget to the {\sc or} gadget
  as shown in Fig.~\ref{fig:3SatGadget}:
   place the variable gadget above the clause
  so as to align the vertical output and input channels,
  and join them into a common channel.
  To make room for the three variable gadgets, we simply extend the black
  areas separating the three input channels in the {\sc or} gadget.
  The unused output channel of each variable gadget 
  is connected to a waste receptacle. Any particle reaching that end
  cannot return to the logic.

  If any input channel of the {\sc or} gadget has a particle,
  then it can reach the output port by the move sequence $\langle d,\ell,d,r \rangle$.
  Furthermore, because variable gadgets place all particles on their
  output ports at the same time, if more than one particle reaches the
  {\sc or} gadget, they will move in unison as drawn in
  Fig.~\ref{fig:3InputOR}, and only one can make it to the output port;
  the others will be stuck in the ``waste'' row, even if extra $\langle \ell,r,u,d \rangle$
  commands are interjected into the intended sequence.
  Hence, a single particle can reach the output of a clause if and only if
  that clause (i.e., at least one of its literals) is satisfied by
  the variable assignment.

\newcommand{\minputANDgate}[1]{
\begin{overpic}[height=0.15\columnwidth]{#1}
\scriptsize
\put(28,47){$o_1$}
\put(37,47){$o_2$}
\put(46,47){$o_3$}
\put(55,47){$o_4$}
\put(64,47){$o_5$}
\put(77,35){waste}
\tiny
\put(31,5){\textcolor{white}{\sc Target}}
\end{overpic}
}

\begin{figure}[htbp]
\centering
\subfloat[][\label{fig:3InputOR} 
 3-input {\sc or}\\$\langle d,\ell,d,r\rangle$ ]
{\begin{overpic}[height=0.15\columnwidth]{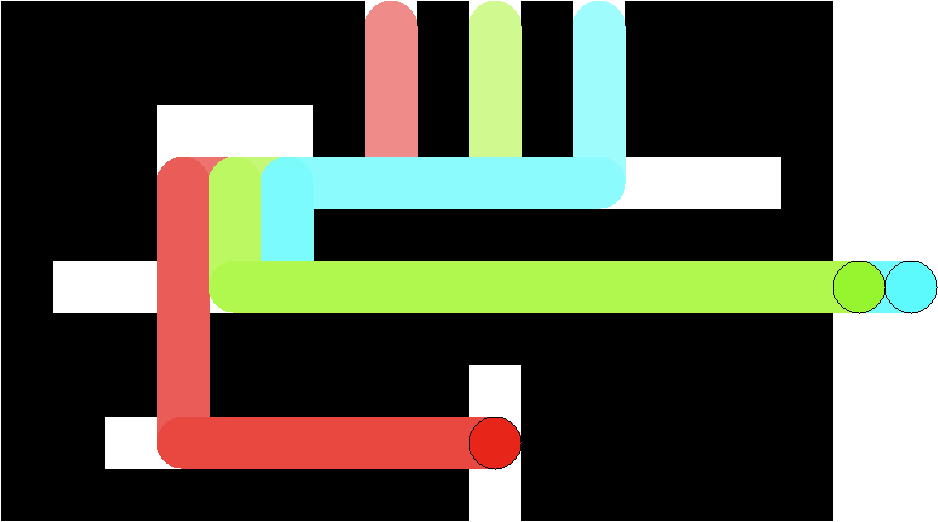}
\scriptsize
\put(38,57){$x_1$}
\put(47,57){$\neg x_3$}
\put(61,57){$x_4$}
\put(50,-3){$o_i$}
\put(90,30){waste}
\end{overpic}
}
\hspace{.1em}
\subfloat[][\label{fig:mInputANDsuccess} 
 $5$-input {\sc and} ({\sc True}) \\$\langle d,\ell,d,r \rangle$  ]
{\minputANDgate{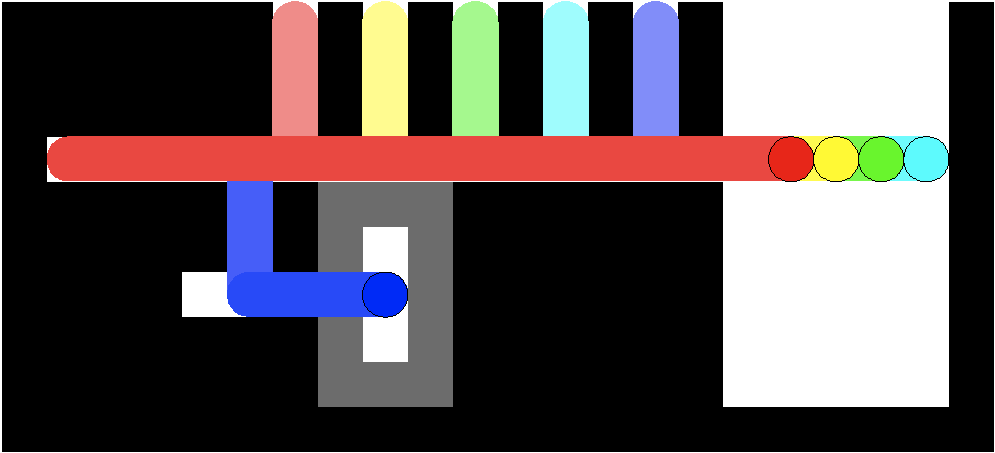}}
\hspace{.1em}
\subfloat[][\label{fig:mInputANDfail}
 $5$-input {\sc and} ({\sc  False}) \\$\langle d,\ell,d \rangle$ ]
{\minputANDgate{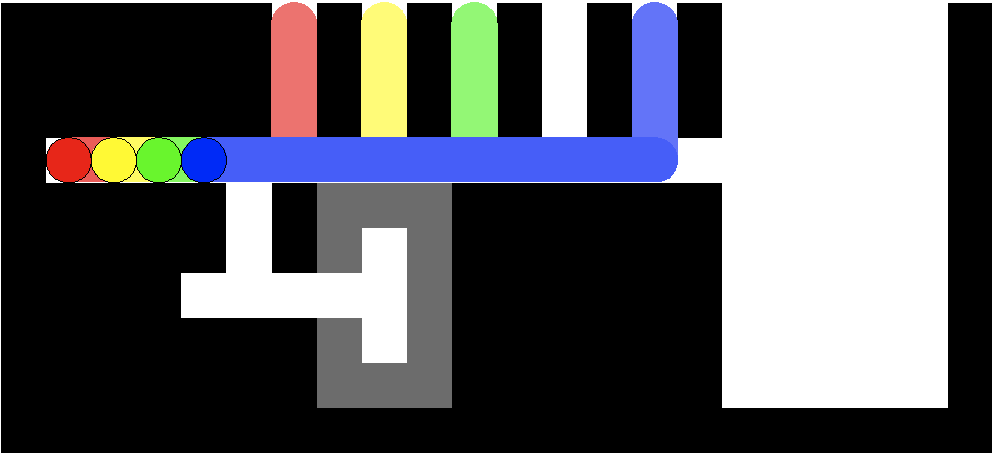}}
\caption{ \label{fig:3InputORandminputAnd} 
Gadgets that use the cycle $\langle d,\ell,d,r \rangle$.
The 3-input {\sc or} gadget outputs one particle if at least one particle enters in an input line and sends any extra particle(s) to a waste receptacle.  The 5-input {\sc and} gadget outputs one particle to the  {\sc Target Location}, marked in gray, if at least 5 inputs are {\sc True}.  Excess particles are sent to a waste receptacle.
 }
\end{figure}

\begin{figure}[htbp]
\subfloat[][\label{fig:3SatStart} \centering
Initial state with particles (colored) on the upper right.  

The objective is to move one particle into the grey target rectangle at lower left.]
{\begin{overpic}[width=1.0\columnwidth]{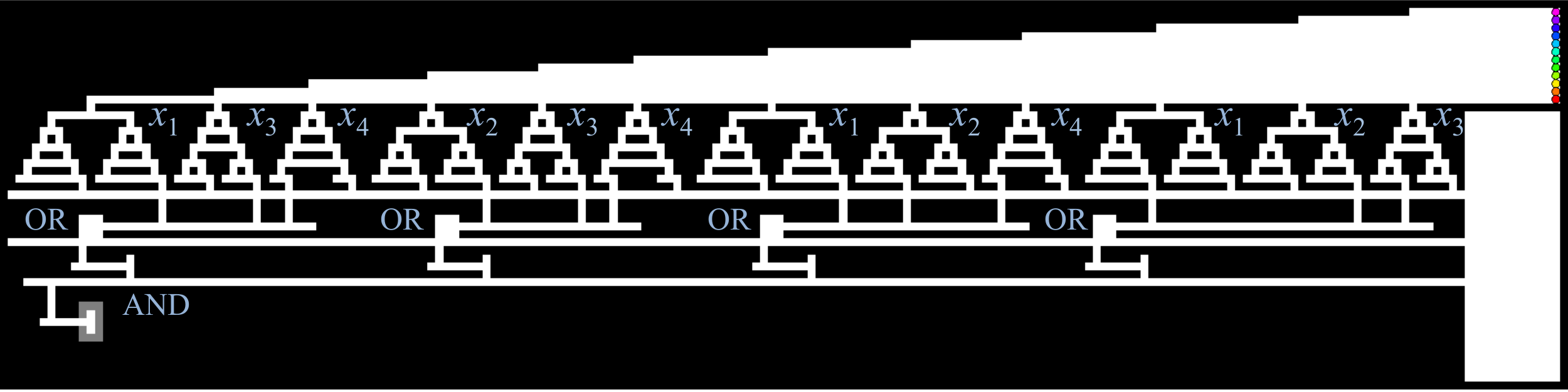}
\end{overpic}
}

\subfloat[][\label{fig:3SatOrFail} 
 Setting variables to ({\sc False, True, False, True}) does not satisfy this 3SAT instance. ]
{\begin{overpic}[width=1.0\columnwidth]{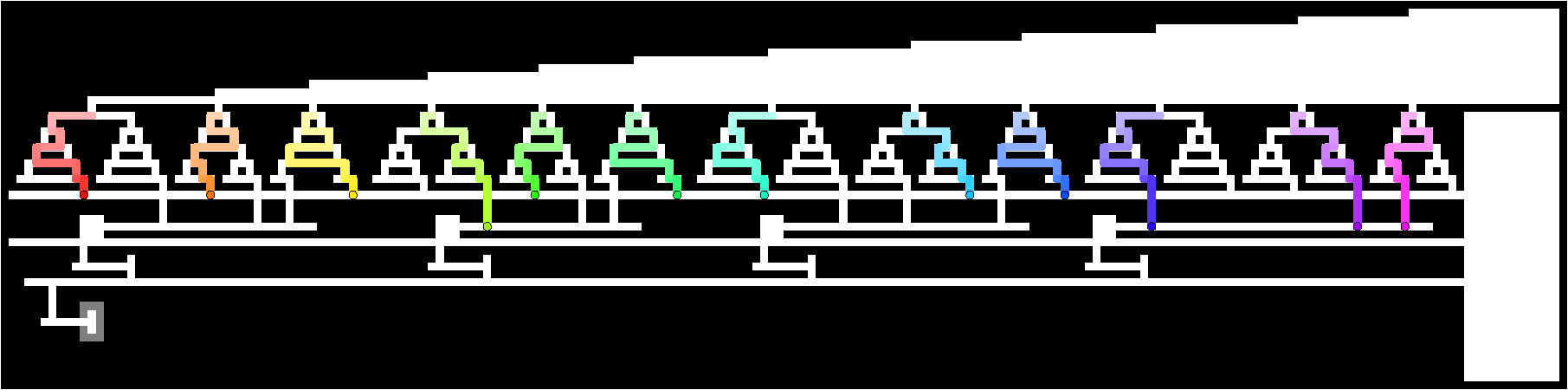}
\end{overpic}
}

\subfloat[][\label{fig:3SatOrSuccess} 
  Setting the variables ({\sc True, False, False, True}) satisfies this 3SAT instance.  ]
{\begin{overpic}[width=1.0\columnwidth]{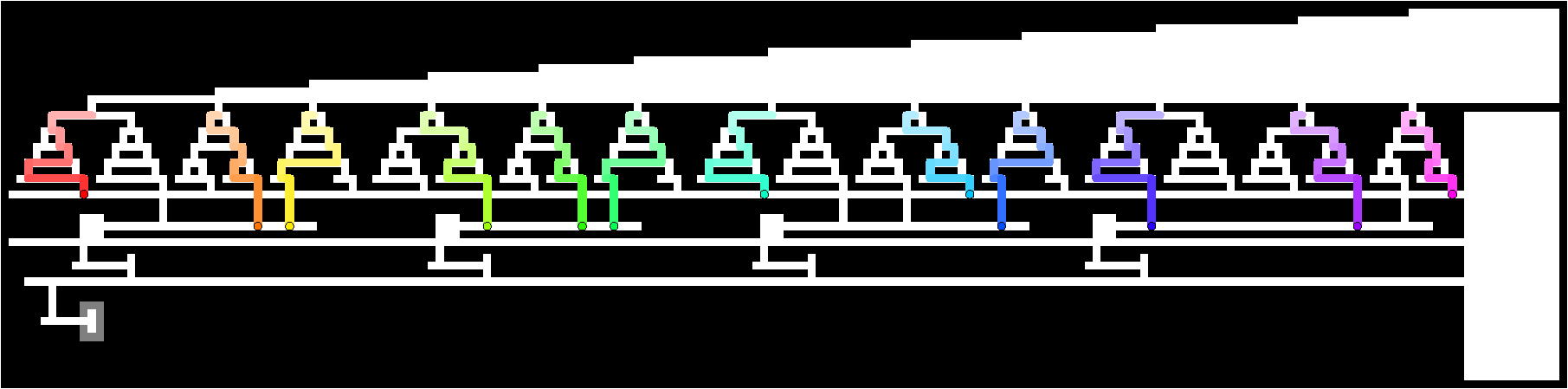}
\end{overpic}
}

\subfloat[][\label{fig:3SatAndSuccess} 
 Successful outcome. ({\sc True, False, False, True}) moves a single particle into the target region. ]
{\begin{overpic}[width=1.0\columnwidth]{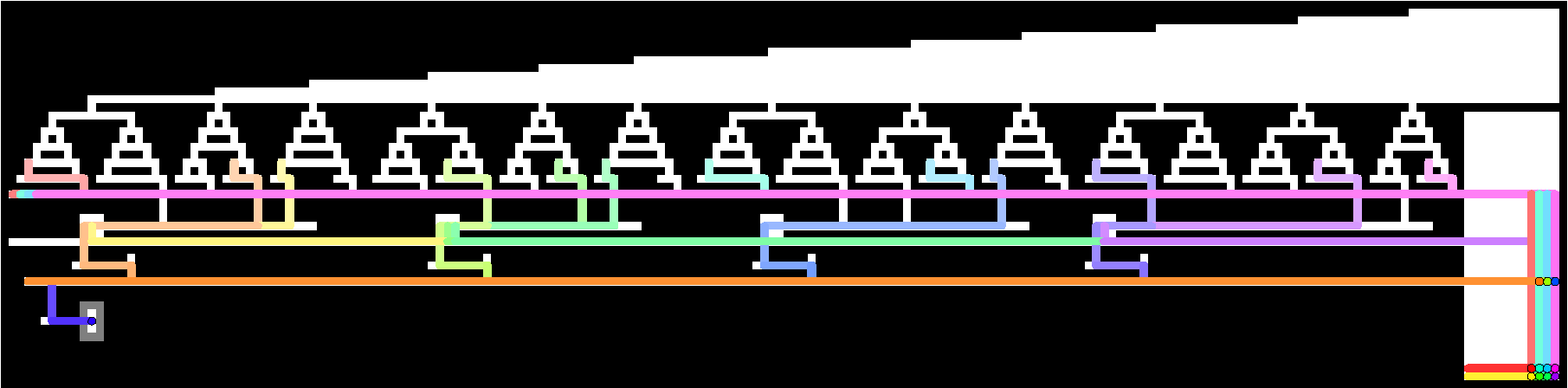}
\end{overpic}
}
\caption{ \label{fig:3SatGadget}
Combining 12 variable gadgets, four 3-input  {\sc or} gadgets, and a 4-input {\sc and} gadget to realize the 3SAT expression $(\neg x_1\vee \neg x_3\vee x_4) \wedge (\neg x_2\vee \neg x_3\vee x_4) \wedge  (\neg x_1 \vee x_2\vee x_4) \wedge (x_1\vee \neg x_2\vee x_3) $.
 }
\end{figure}

\paragraph{Check gadget}
  As the final stage of the computation, we check that all clauses were simultaneously satisfied
  by the variable assignment, using the $m$-input {\sc and} gadget
  shown in Fig.~\ref{fig:mInputANDsuccess} and \ref{fig:mInputANDfail}.
  Specifically, we place the clause gadgets along a horizontal line
  and connect their vertical output channels to the vertical input channels
  of the check gadget.  Again we can align the channels by extending the
  black areas that separate the input channels of the {\sc and} gadget, as shown in the composite diagram Fig.~\ref{fig:3SatGadget}.

  The intended solution sequence for the {\sc and} gadget is $\langle d, \ell, d, r \rangle$.
  The {\sc and} gadget is designed with the downward channel exactly
  $m$ units to the right from the left wall 
  and more than $2m$ units from the right wall, so for any particle to reach the
  downward channel (and ultimately, the target location), at least $m$ particles 
  must be presented as input.  Because each input channel will present at most
  one particle (as argued in a clause), a particle can reach the final destination
  if and only if all $m$ clauses output a particle, which is possible exactly when
  all clauses are satisfied by the variable assignment.

Clearly, the size of all parts of the construction is polynomial in the size of the original 3SAT instance. This completes the reduction and the NP-hardness proof.

\end{proof}


We conjecture that {\sc GlobalControl-ManyParticles} is in fact {\sc pspace}-complete.
One approach would be to simulate
nondeterministic constraint logic \cite{HEARN200572in2005},
perhaps using a unique move sequence of the form $\langle d$, $\ell/r$, $d$, $\ell/r$,
$\dots \rangle$ to identify and ``activate'' a component.
One challenge is that all gadgets must properly reset to
their initial state without permanently trapping any particles. 
However, we are able to prove that a variant of this problem is {\sc pspace}-complete in Section \ref{subsec:pspaceComplete}.


\section{Matrices}\label{sec:matrices}
The previous section investigated pathologically difficult configurations. This
section investigates a complementary problem. Given the same particle and world
constraints as before, 
 what types of control are possible and economical if we are free to design the environment?

First, we describe an arrangement of obstacles that implement an arbitrary
matrix permutation in four commands.  Then we provide efficient algorithms for
sorting matrices. We finish with potential applications.

%

\subsection{A Workspace for a Single Permutation}
\label{subsec:single}

For our purposes, a \emph{matrix} is a 2D array of particles 
(each possibly of a different color).
For an $a_r \times a_c$ matrix $A$ and a $b_r \times b_c$ matrix $B$,
of equal total size $N = a_r a_c = b_r  b_c$,
a \emph{matrix permutation} assigns each element in $A$
a unique position in~$B$.
Example constructions that execute matrix permutations of
total size $N=15$ and $100$ are shown, respectively, in Fig.~\ref{fig:MatrixPermuteAI} and \ref{fig:MatrixPermuteRD}. 

\begin{figure}
\centering
\href{http://www.youtube.com/watch?v=3tJdRrNShXM}{
\begin{overpic}[width=\columnwidth]{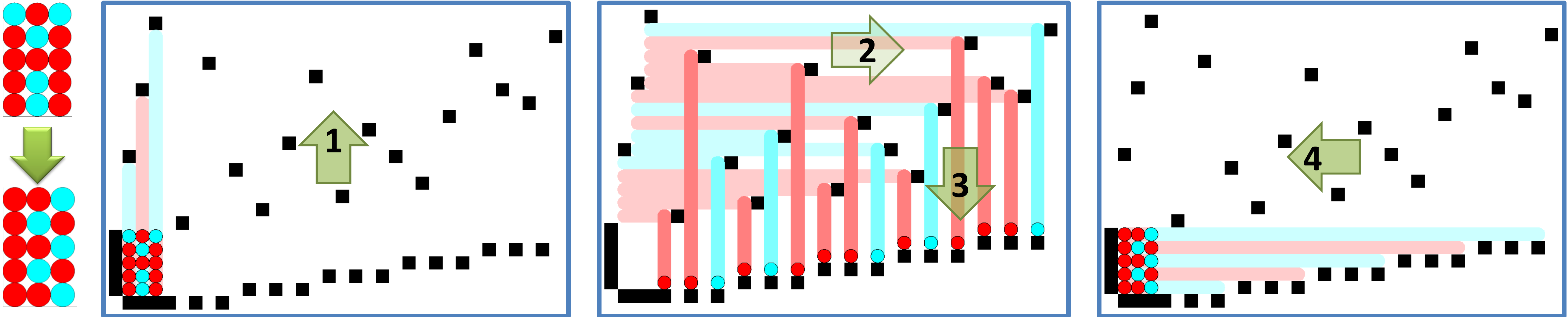}
\end{overpic}}
\caption{
\label{fig:MatrixPermuteAI}
\href{http://www.youtube.com/watch?v=3tJdRrNShXM}{
In this image for $N=15$, black cells are obstacles, white cells are free, and colored discs are individual particles. The world has been designed to permute the particles between `A' into `B' every four steps: $\langle u,r,d,\ell \rangle$. See video at http://youtu.be/3tJdRrNShXM. Visually, the distinction between particles of the same color does not matter; however, the arrangement of obstacles induces a specific permutation of individual particles.}
}
\end{figure}

\begin{figure}
\centering
\href{http://www.youtube.com/watch?v=3tJdRrNShXM}{
\begin{overpic}[width=\columnwidth]{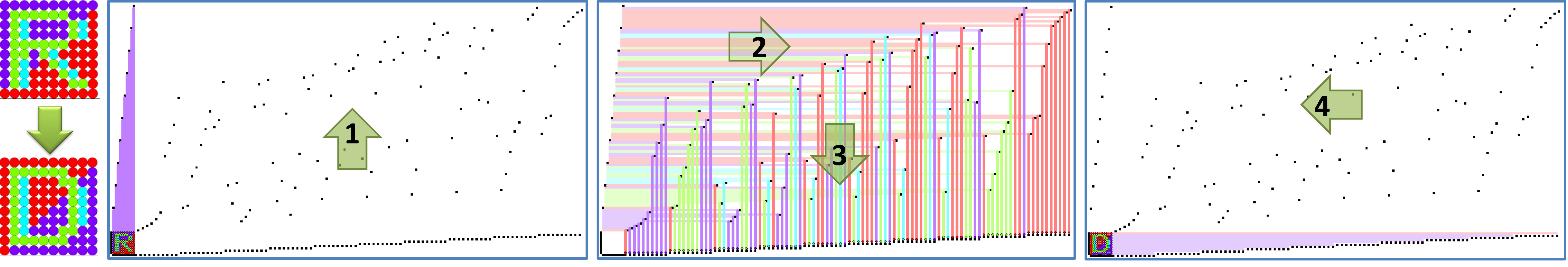}
\end{overpic}}
\caption{
\label{fig:MatrixPermuteRD}
In this larger example with $N=100$, the different control sections are easier to see than in Fig.~\ref{fig:MatrixPermuteAI}.  (1) The staggered obstacles on the left spread the matrix vertically, (2) the scattered obstacles on the upper right permute each element, and (3) the staggered obstacles along the bottom reform each row, which are collected by  (4). The cycle resets every 740 iterations.
See \url{http://youtu.be/eExZO0HrWRQ} for an animation of this gadget.  
}
\end{figure}

\begin{theorem} \label{thm:ArbPermutationUsingObstacles}
  Let $A$ and $B$ be matrices with dimensions as above. Any matrix permutation that transforms $A$ into $B$ can be executed by a set of obstacles 
  in just four moves.
  For $N$ particles, the constructed arrangement of obstacles requires $(3N+1)^2$ space and $4N+1$ obstacles. If particles move with a speed of $v$, the required time for those four moves is $12N/v$.
  and in time $12N/\mathrm{speed}$. 
\end{theorem}

\begin{proof}
The reader is referred to Fig.~\ref{fig:MatrixPermuteAI} and \ref{fig:MatrixPermuteRD}
for examples.  The move sequence is $\langle u,r,d,\ell \rangle$. The lower-left particle starts at $(0,0)$.

{\bf Move~1:} We place $a_c$ obstacles, one for each column of $A$, spaced vertically $a_r-1$ units apart, such that moving $u$ spreads the particle array into a staggered vertical line. Each particle now has its own row.
{\bf Move~2:} We place $N$ obstacles to stop each particle during the move~$r$.   Because each particle has a unique row, it can be stopped at any arbitrary column by its obstacle. We leave an empty column between each obstacle to prevent collisions during the next move.
{\bf Move~3:} We place $N$ particles to stop particles during the move $d$, which arranges the particles in their final rows.  These rows are spread in a staggered horizontal line.  
{\bf Move~4:} We place $a_r$ obstacles in a vertical line from $(-1,1)$ to $(-1,a_c)$. Moving $\ell$ stacks the staggered rows into the desired permutation and returns the array to the initial position.
\end{proof}


By reapplying the same permutation enough times, we can return to the original configuration.  The permutations shown in Fig.~\ref{fig:MatrixPermuteAI} return to the original image in two cycles, while Fig.~\ref{fig:MatrixPermuteRD} requires 740 cycles.  In fact, any permutation of $N$ elements will return to its original position after it is repeated a finite number of times \cite{DummitFoote}. 

For a two-color image, we can always construct a permutation that resets in 2 cycles. 
If the matrix has only two colors then for each entry in the matrix a permutation of the particles will either flip the color or keep the color constant in that given entry. If the permutation flips the color of a particular entry, then doing the permutation twice will flip this entry back to its original color. If the permutation keeps the color of a particular entry constant, then doing the permutation twice will also preserve the color of that entry. Performing the permutation twice always results in the original matrix. Such a permutation is an \emph{involution}, a function that is its own inverse. An involution often does not exist for permutations on images with more than two colors.


\subsection{A Workspace for Arbitrary Permutations}
\label{subsec:arbitrary}

Theorem~\ref{thm:ArbPermutationUsingObstacles} can be exploited to generate larger sets
of permutations or even all possible permutations. 
There is a tradeoff between the number
of introduced obstacles and the number of moves required for realizing a permutation.

We start with obstacle sets that require only a small number of moves.

\begin{theorem}
\label{thm:kPermutationsInLogkMoves}
      For an arbitrary set of $k$ permutations of $N$ particles, we can construct
      a set of $O(kN)$ obstacles, such that we can switch from a start arrangement
      into any of the $k$ permutations using at most $O(\log k)$ force-field moves.
\end{theorem}
\begin{proof} See Fig.~\ref{fig:ChooseKPermutations}.
      Build a binary tree of depth $\log k$ for choosing between the permutations
      by a sequence of $\langle r$, $d$, $(r/\ell)$, $d$, $(r/\ell)$, \dots,
$d$, $(r/\ell)$, $d$, $\ell$, $u \rangle$ with $\log k$ decisions between $r$ and $\ell$,
from the initial prefix $\langle r,d \rangle$ to the final suffix $\langle
d,\ell,u \rangle$.  This gets the particles to the set of obstacles for
performing the appropriate permutation.
\end{proof}

\begin{figure}[t]
\centering
\href{http://www.youtube.com/watch?v=3tJdRrNShXM}{
\begin{overpic}[width=1.0\columnwidth]{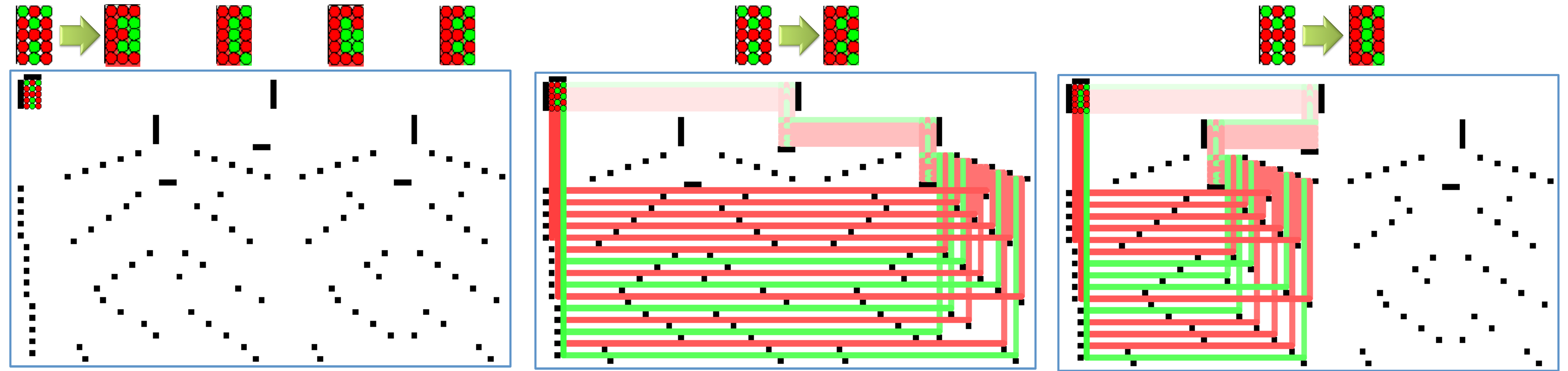}
\end{overpic}}
\caption{
\label{fig:ChooseKPermutations}
\href{http://www.youtube.com/watch?v=3tJdRrNShXM}{For any set of $k$ fixed, but arbitrary permutations of $N$ particles, we can construct
      a set of $O(kN)$ obstacles, such that we can switch from a start arrangement
      into any of the $k$ permutations using at most $O(\log k)$ force-field moves. Here $k=4$ and `A' is transformed into `B', C', `D', or `E' in eight moves: $\langle r,d,(r/\ell),d,(r/\ell),d,\ell,u \rangle$.}
}
\end{figure}


\begin{corollary}
\label{cor:NfactObstacles}
      For any $N$ and an arbitrary but fixed $\varepsilon > 0$, we can construct a set of $(N!)^\varepsilon$ obstacles such that any permutation of $N$ particles
      can be achieved by at most $O(N \log N)$ force-field moves.
\end{corollary}
\begin{proof} This follows from Theorem \ref{thm:kPermutationsInLogkMoves}. With $k=(N!)^{\varepsilon}/N$, $\log k$ becomes $\varepsilon \log (N!)- \log N$, i.e., $O(N\log N)$.
\end{proof}

Now we proceed to more economical sets of obstacles, with arbitrary permutations realized by
clockwise and counterclockwise move sequences. We make use of the following easy lemma, 
which shows that two base permutations are enough to generate any desired rearrangement.

\begin{lemma}\label{lemma:TwoPermutationsGeneratesUniversalPermutations}
      Any permutation of $N$ objects can be generated by the
      two base permutations $p=(1,2)$ and $q=(1,2,\ldots, N)$.
      Moreover, any permutation can be generated by a sequence
      of length at most $N^2$ that consists of $p$ and $q$.
\end{lemma}

The proof is elementary and left to the reader.
This allows us to establish the following result.

\begin{theorem}
\label{thm:NsqMovesToSort} 
      We can construct a set of $O(N)$ obstacles such that
      any $a_r\times a_c$ arrangement of $N$ particles can be rearranged into
      any other $a_r\times a_c$ arrangement $\pi$ of the same particles, using
      at most $O(N^2)$ force-field moves.
\end{theorem}
\begin{proof}       See Fig.~\ref{fig:BubbleSort}.
      Use Theorem \ref{thm:ArbPermutationUsingObstacles} to build two sets of obstacles, one each for $p$ and $q$,
      such that $p$ is realized by the sequence $\langle u,r,d,\ell \rangle$ (clockwise)
      and $q$ is realized by $\langle r,u,\ell,d \rangle$ (counterclockwise).
      Then we use the appropriate sequence for generating $\pi$ in $O(N^2)$ moves.
\end{proof}

\begin{figure}[t]
\centering
\begin{overpic}[width=0.95\columnwidth]{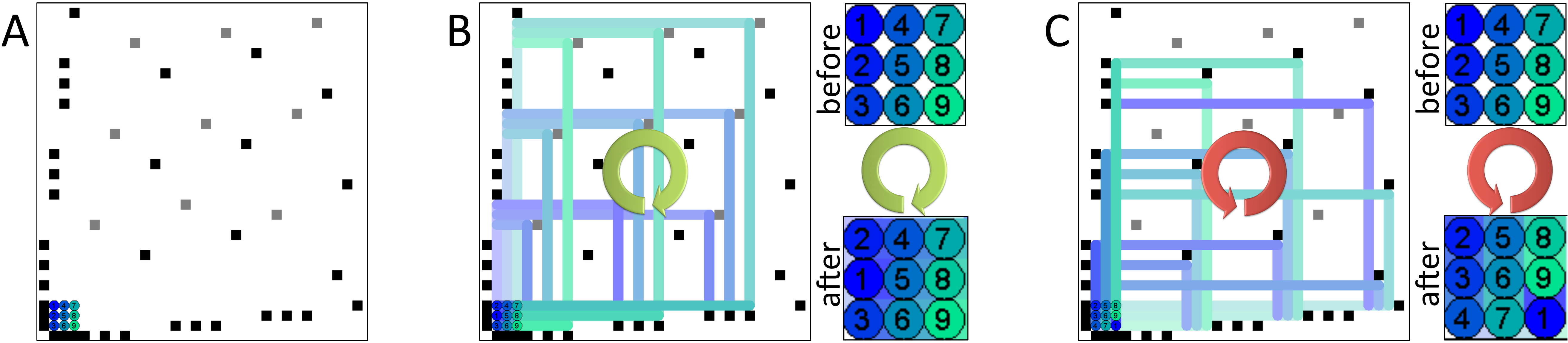}
\end{overpic}
\caption{
\label{fig:BubbleSort}
Repeated application of two base permutations can generate any permutation, when used in a manner similar to  {\sc Bubble Sort}.  The obstacles in (Fig.~\ref{fig:BubbleSort}A) generate the base permutation $p=(1,2)$ in the clockwise direction $\langle u,r,d,\ell \rangle$ (Fig.~\ref{fig:BubbleSort}B) and $q=(1,2,\ldots, N)$ in the counter-clockwise direction $\langle r,u,\ell,d \rangle$ (Fig.~\ref{fig:BubbleSort}C).
}
\end{figure}

%

Using a larger set of base permutations allows us to reduce the number of necessary
moves. Again, we make use of a simple base set for generating arbitrary permutations.

\begin{lemma}\label{lemma:NbasePermutations}
      Any permutation of $N$ objects can be generated by the
      $N$ base permutations $p_1=(1,2), p_2=(1,3),\dots, p_{N}=(1,(N-1))$ and $q=(1,2,\ldots, N)$.
      Moreover, any permutation can be generated by a sequence
      of length at most $N$ that consists of $p_i$ and $q$.
      \end{lemma}

The proof is again completely straightforward and left to the reader.

\begin{theorem}\label{thm:NlogNmovesWithN2Obstacles}
      We can construct a set of $O(N^2)$ obstacles such that
      any $a_r \times a_c$ arrangement of $N$ particles can be rearranged into
      any other $b_r \times b_c$ arrangement $\pi$ of the same particles, using
      at most $O(N \log N)$ force-field moves.
      \end{theorem}
\begin{proof}   
 Use Theorem \ref{thm:ArbPermutationUsingObstacles} to build $N$ sets of obstacles, one each for $p_1,\ldots,p_{N-1}, q$.
      Furthermore, use  Lemma~\ref{lemma:NbasePermutations} for generating all permutations with at most
      $N$ different of these base permutations, and Theorem~\ref{thm:kPermutationsInLogkMoves} for switching between
      these $k=N$ permutations. Then we can get $\pi$ with at most $N$ cycles, each consisting
      of at most $O(\log N)$ force-field moves.
      \end{proof}

This is the best possible with respect to the number of moves in the following sense:

\begin{theorem}\label{thm:PermutationOptimalNumMoves}
       Suppose we have a set of obstacles such that any permutation of a rectangular arrangement
       of $N$ particles can be achieved by at most $M$ force-field moves. Then $M$ is at least $\Omega(N \log N)$.
\end{theorem}
\begin{proof}
       Each permutation must be achieved by a sequence of force-field moves.
       Because each decision for a force-field move $\langle u,d,\ell,r  \rangle$
       partitions the remaining set of possible permutations into at most four different subsets,
       we need at least $\Omega(\log(N!))= \Omega(N \log N)$ such moves.
\end{proof}

\subsection{{\sc pspace}-Completeness}
\label{subsec:pspaceComplete}

In Section~\ref{sec:mazes}, we showed that the problem {\sc GlobalControl-ManyParticles} is com\-pu\-ta\-tio\-nal\-ly intractable in a particular
sense: given an initial configuration of movable particles and fixed obstacles, it is
NP-hard to decide whether {\em any} individual particle can be moved to a specified location. 
In the following, we show that minimizing the number of moves for achieving a desired goal configuration for {\em all} particles is {\sc pspace}-complete.

\begin{theorem}
  Given an initial configuration of (labeled) movable particles and fixed obstacles, it is
  {\sc pspace}-complete to compute a shortest sequence of force-field moves to achieve another (labeled) configuration.
\end{theorem}

\begin{proof}
The proof is largely based on a complexity result by Jerrum~\cite{j-cfmlg-85}, who considered the following problem:
Given a permutation group specified by a set of generators and a single target permutation $\pi$, which is
a member of the group, what is the shortest expression for the target permutation in terms of the generator? This problem
was shown to be {\sc pspace}-complete in~\cite{j-cfmlg-85}, even when the generator set consists of only two permutations, say, $\pi_1$ and $\pi_2$.

As shown in Subsection~\ref{subsec:arbitrary},  we can realize any matrix permutation $\pi_i$ of a rectangular arrangement of
particles by a set of obstacles, such that this permutation $\pi_i$ is carried out by a quadruple of force-field moves.
We can combine the sets of obstacles for the two different permutations $\pi_1$ and $\pi_2$, such that $\pi_1$
is realized by going through a clockwise sequence $\langle u, r, d, \ell\rangle$, while $\pi_2$ is realized by a counterclockwise
sequence $\langle r, u, \ell, d\rangle$. We now argue that a target permutation $\pi$ of the matrix can be realized by
a minimum-length sequence of $m$ force-field moves if and only if $\pi$ can be decomposed into a sequence of
a total of $n$ applications of permutations $\pi_1$ and $\pi_2$, where $m=4n$.

The ``if'' part is easy: simply carry out the sequence of $n$ permutations, each realized by a (clockwise or counterclockwise)
quadruple of force-field moves. For the ``only if'' part, suppose we have a shortest sequence of $m$ force-field moves to achieve permutation
$\pi$, and consider an arbitrary subsequence that starts from the base position in which the particles form a rectangular arrangement
in the lower left-hand corner. It is easy to see that a minimum-length sequence cannot contain two consecutive moves that are
both horizontal or both vertical: these moves would have to be be in opposite directions, and we could shorten the sequence by omitting
the first move.
Furthermore, by construction of the obstacle set, the first move must be $u$ or $r$. Now it is easy to check that
the choice of the first move determines the next three ones: $u$ must be followed by $\langle r, d, \ell\rangle$; similarly,
$r$ must be followed by $\langle u, \ell, d\rangle$. Any other choice for moves 2--4 would produce a longer overall sequence
or destroy the matrix by leading to an arrangement from which no recovery to a rectangular matrix is possible. Therefore, the overall sequence
can be decomposed into $m=4n$ clockwise or counterclockwise quadruples. As described, each of these quadruples represents either
$\pi_1$ or $\pi_2$, so $\pi$ can be decomposed into $n$ applications of permutations $\pi_1$ and $\pi_2$.
This completes the proof.
\end{proof}

Note that the result also implies the existence of solutions of exponential length, which can occur with polynomial space.
Binary counters are particular examples of such long sequences that are useful for many purposes.

\section{Limitations of Particle Logic}\label{sec:logic}

After considering the complexity of rearranging given arrangements of
particles, i.e., {\em external computation}, we now turn to using the particles
themselves for performing logic operations, i.e., {\em internal computation}.
We first establish the limitations of 1$\times$1 particles; details on designing the full range of
logic gates with the help of 2$\times$1 particles are described in the following Section~\ref{sec:Design}.

\subsection{Dual-Rail Logic and {\sc fan-out} Gates}
 In Section \ref{sec:mazes} we showed that given only obstacles and particles that move maximally in response to an input, we can construct a variety of logic elements.  
 These include variable gadgets that enable setting multiple copies of up to $n$ variables to be {\sc True} or {\sc False} (Fig.~\ref{fig:VariableGadget}) and 
   $m$-input {\sc or} and {\sc and} gates (Fig.~\ref{fig:3InputORandminputAnd}). 
   Unfortunately, we cannot build  {\sc not} gates because our system of particles and obstacles is conservative---we cannot create a new particle at the output when no particle is supplied to the input. A  {\sc not}  gate is necessary to construct a logically complete set of gates.  To do this, we rely on a form of \emph{dual-rail logic}, where both the state  and inverse ($A$ and $\bar{A}$) of each signal are propagated throughout the computation.  Dual-rail logic is often used in low-power electronics to increase the signal-to-noise ratio without increasing the voltage \cite{zimmermann1997low}.  With dual-rail logic we can now construct the missing  {\sc not} gate, as shown in Fig.~\ref{fig:ParticleLogic11}. The command sequence $\langle d,\ell,u,r\rangle$ inverts the input.  
   

\begin{figure}
\centering
\begin{overpic}[width =.85\columnwidth]{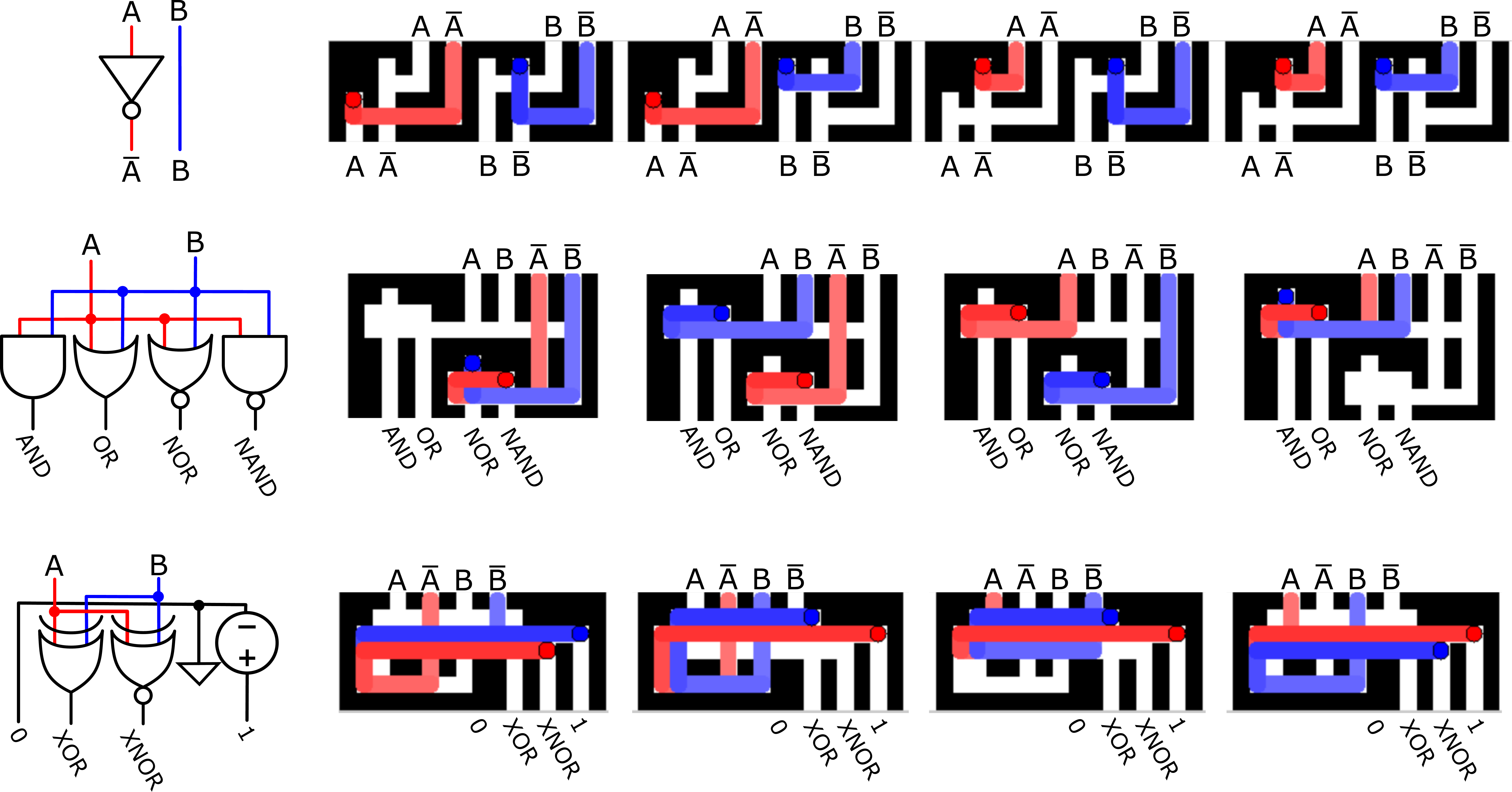}\end{overpic}
\caption{\label{fig:ParticleLogic11}Schematic and diagram of  dual-rail logic gates. Each gate employs the same clock sequence $\langle d,\ell,u,r \rangle$, the four inputs correspond to $A,\bar{A},B,\bar{B}$. 
The top row is a {\sc not} gate and a connector. 
The middle row is a universal logic gate whose four outputs are {\sc and, nand, or, nor}.  
The bottom row gate outputs the {\sc xor, xnor} of the inputs and constants 1 and 0.  \href{http://youtu.be/mJWl-Pgfos0}{See video at \url{http://youtu.be/mJWl-Pgfos0} for a hardware demonstration.}
}
\vspace{-1em}
\end{figure}

 We now revisit the  {\sc or} and {\sc and} gates
of~Fig.~\ref{fig:3InputORandminputAnd}, using dual-rail logic and the four
inputs $A,\bar{A},B,\bar{B}$.  The gate in the middle row of
Fig.~\ref{fig:ParticleLogic11} can simultaneously compute  {\sc and, or, nor,}
and {\sc nand},  using the same  command sequence $\langle d,\ell,u,r\rangle$ as
the {\sc not} gate.  Outputs can be piped for further logic using the
interconnections in Fig.~\ref{fig:ParticleLogic11}. Unused outputs can be piped
into a storage area and recycled for later use.


  Dual-rail devices open up new opportunities, including  {\sc xor} and {\sc xnor} gates, which are not conservative using single-rail logic.  This gate, shown in the bottom row of~Fig.~\ref{fig:ParticleLogic11} also outputs a constant 1 and 0.

    Consider the half adder shown in Fig.~\ref{fig:HalfAdder}.     With an {\sc and} and {\sc xor} we can compactly construct a half adder.  We are hindered by an inability to construct a {\sc fan-out} device. 
  The \emph{fan out} of a logic gate output is the number of gate inputs it can feed or connect to.  In the particle logic demonstrated thus far, each logic gate output could fan out to only one gate.  This is sufficient for \emph{sum of products} and \emph{product of sums}  operations in CPLDs (complex programmable logic devices) but insufficient for more flexible architectures.
 Instead, we must take any logical expression and create multiple copies of each input.  For example, a half adder requires only one {\sc xor} and one  {\sc and} gate, but our particle computation requires two $A$  and two $B$ inputs.
  In the rest of this section we prove the insufficiency of unit-sized particles for the implementation of {\sc fan-out} gates and design a {\sc fan-out} gate using $2\times 1$ particles. 
  
     \begin{figure}
   \centering
\begin{overpic}[width =0.35\columnwidth]{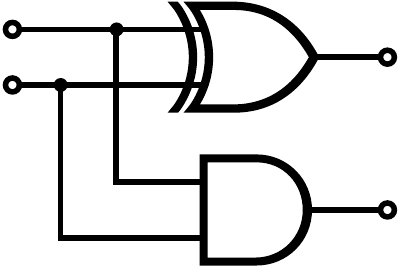}
\put(-7,59){$A$}
\put(-7,44.5){$B$}
\put(100,51.7){$S$}
\put(100,13.5){$D$}
\end{overpic}
\caption{
\label{fig:HalfAdder}
The half adder shown above requires two copies of  $A$ and  $B$.
}
\vspace{-1em}
\end{figure}
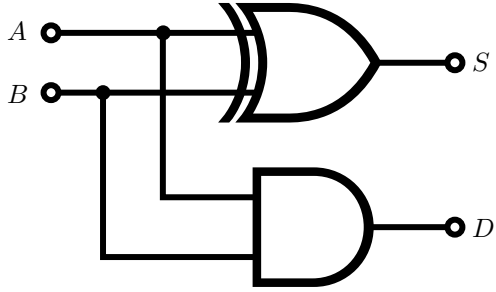

\subsection{Only 1$\times$1 Particles Are Insufficient}\label{sec:dualnec}

First we provide terminology to define how particles interact with each other. We say
that particle $q$ \emph{blocks} particle $p$ during a 
move $m_k$, if $p$ is prevented from reaching location $\mathbf{x}=(x,y)$
because particle $q$ occupies this location. As a consequence, at the end
of $m_k$, $q$'s location is $\mathbf{x}$, while particle $p$'s location is adjacent to
$\mathbf{x}$, depending on the direction of $m_k$.
Furthermore, the sequence of locations $\langle \mathbf{s},\ldots,\mathbf{g} \rangle$
of a particle $p$ from a start location $\mathbf{s}$ to its goal location $\mathbf{g}$
describes its \emph{path}. For an unchanged sequence of force-field moves 
and obstacles, a particle $p$ can only be prevented from reaching its destination by
adding additional particles, or removing existing ones. We argue in the following
that this will still lead to $g$ being occupied at the end of the sequence, possibly
by a different particle.

%


\begin{lemma}\label{thm:AdditionalParticlesCannotPreventAnOccupation}
If given a fixed workspace $W$\! and a command sequence $\mathbf{m}$ that moves a particle $p$ from start location $\mathbf{s}$ to goal location $\mathbf{g}$, then adding additional particles anywhere in $W$\! at any stage of the command sequence cannot prevent $\mathbf{g}$ from being occupied at the conclusion of sequence $\mathbf{m}$.
\end{lemma}

\begin{proof} Consider the effect of adding a particle to workspace $W$\!. If
$p$ never gets blocked by any particle $q$, then 
$p$'s path remains the same.
Therefore, at the conclusion of $\mathbf{m}$, $p$ occupies $\mathbf{g}$. 

Now suppose $p$ gets blocked by $q$. By the definition of blocking, $q$
prevents $p$ from reaching some location $\mathbf{x}$ because $q$ already
occupies this location. After the blocking, the command sequence will continue and
so particle $q$ will continue on $p$'s original path, following the same
instructions and therefore ending up in the same location, $\mathbf{g}$, unless
$q$ gets blocked by yet another particle. By induction, additional particles will have the
same effect. If $q$ gets blocked by any other particle, then this particle will continue
on $p$'s original path. Thus by adding more particles, it is impossible to
prevent some particle from occupying $\mathbf{g}$ at the conclusion of
$\mathbf{m}$.  
\end{proof}

\begin{corollary}
A  {\sc not} gate without dual-rail inputs cannot be constructed.
\end{corollary}
\begin{proof}
By contradiction.
A particle logic {\sc not} gate without dual-rail inputs has one input at $\mathbf{s}$, one output at $\mathbf{g}$, an arbitrary (possibly zero) number of asserted inputs, which are all initially occupied, and an arbitrary (possibly zero) number of waste outputs.

To satisfy the {\sc not} gate conditions given a command sequence $\mathbf{m}$, the following conditions must be satisfied.
\begin{enumerate}
\item If $\mathbf{s}$ is initially unoccupied, $\mathbf{g}$ must be occupied at the conclusion of $\mathbf{m}$.
\item If $\mathbf{s}$ is initially occupied, $\mathbf{g}$ must be unoccupied at the conclusion of $\mathbf{m}$.
\end{enumerate}
By Lemma~\ref{thm:AdditionalParticlesCannotPreventAnOccupation}, if $\mathbf{s}$ initially unoccupied results in $\mathbf{g}$ being occupied by some particle $p$ at the conclusion of $\mathbf{m}$, then the addition of a particle $q$ at $\mathbf{s}$ cannot prevent $\mathbf{g}$ from being filled, resulting in a contradiction.
\end{proof}

This shows that dual-rail logic is necessary for the formation of {\sc not} gates. 

Additionally, we show that $1\times1$ particles are insufficient to produce {\sc fan-out} gates. To this end, we must examine the possibilities both when we add additional particles to the scenario and when we remove them.  

\begin{lemma}\label{thm:TwoParticlesTwoGoalsImpliesOneParticleOneGoal} 
Consider a given workspace $W$\! with a number of particles, two of which are 
$p_1$ and $p_2$, initially at $\mathbf{s}_1$ and $\mathbf{s}_2$.
Let $\mathbf{m}$ be a command sequence that moves $p_1$ and $p_2$ to the
respective goal locations $\mathbf{g}_1$ and $\mathbf{g}_2$. Then deleting
either $p_1$ or $p_2$ from the original configuration 
results in at least one of $\mathbf{g}_1$ or $\mathbf{g}_2$ being
occupied at the conclusion of $\mathbf{m}$.  
\end{lemma}

\begin{proof} 
Without loss of generality, suppose we remove particle $p_1$.  

First suppose $p_2$ never gets blocked by $p_1$. Then the removal of $p_1$ will
not affect the path of $p_2$. Particle $p_2$ has the same number of blockings that
it had before the removal of $p_1$ and so 
$p_2$ will follow the same path and
occupy $\mathbf{g}_2$ at the conclusion. 

Alternatively, suppose $p_2$ gets blocked by $p_1$ when $p_1$ is occupying
location $\mathbf{x}$. Because $p_1$ is removed, it no longer occupies
$\mathbf{x}$ during this move; because it was stopped in the common direction
when blocking $p_2$, particle $p_2$ gets stopped by this obstacle at
location $\mathbf{x}$, previously occupied by $p_1$. Particle $p_2$ now
proceeds along the path previously traveled by $p_1$. Effectively, $p_2$ has
replaced $p_1$ and follows the path until it reaches $\mathbf{g}_1$. Successive
blockngs between $p_2$ and $p_1$ in the original scenario are resolved in the same
manner.  
\end{proof}

  \begin{table}
\begin{displaymath}
\begin{array}{ccc|cccc}
\multicolumn{3}{c}{\emph{Inputs}} & \multicolumn{4}{c}{\emph{Outputs}} \\
   A & \overline{A} & 1 & A & A &  \overline{A} & \overline{A}\\
\hline
0 & 1 & 1 & 0 & 0 & 1 & 1  \\
1 & 0 & 1 & 1 & 1 & 0 & 0  \\
\end{array}
\end{displaymath}
\caption{{\sc fan-out} operation. This cannot be implemented with 1$\times$1 particles and obstacles. Our technique uses 2$\times$1 particles. }
  \label{tab:Fanout}
\end{table}

In the context of programmable matter, it is natural to consider systems in which particles
are moved around, but neither created nor destroyed; such a system is called \emph{conservative}.
As it turns out, this has important consequences.

\begin{theorem}\label{cor:No1x1FanOut}
A  conservative dual-logic {\sc fan-out} gate cannot be constructed using only 1$\times$1 particles.
\end{theorem}

\begin{proof} 
We assume such a {\sc fan-out} gate exists and reach a contradiction. 
Consider a  {\sc fan-out} gate $W$,  dual-rail input locations  $\mathbf{s}_{a}$, $\mathbf{s}_{\overline{a}}$, and dual-rail output locations  $\mathbf{g}_{a_1}, \mathbf{g}_{a_2},\mathbf{g}_{\overline{a}_1},\mathbf{g}_{\overline{a}_2}$. Because particle logic is conservative, there must be at least one additional input location $\mathbf{s}_p$ and particle $p$. A {\sc fan-out} gate implements the truth table shown in Table \ref{tab:Fanout}. Given an arbitrary command sequence $\mathbf{m}$:  
\begin{enumerate}
\item If $\mathbf{s}_{a}$  and $\mathbf{s}_p$ are initially occupied and $\mathbf{s}_{\overline{a}}$ vacant at the conclusion of $\mathbf{m}$, then $\mathbf{g}_{a_1}$ and $\mathbf{g}_{a_2}$ are occupied and the locations  $\mathbf{g}_{\overline{a}_1}$ and $\mathbf{g}_{\overline{a}_2}$ are vacant.
\item If $\mathbf{s}_{a}$ is initially vacant and $\mathbf{s}_{\overline{a}}$ and $\mathbf{s}_p$ are occupied at the conclusion of $\mathbf{m}$, then $\mathbf{g}_{a_1}$ and $\mathbf{g}_{a_2}$ are vacant and the locations  $\mathbf{g}_{\overline{a}_1}$ and $\mathbf{g}_{\overline{a}_2}$ are occupied.
\end{enumerate}

We will now assume that condition 1, above, is the original scenario and  add
and subtract particles, applying
Lemmas~\ref{thm:AdditionalParticlesCannotPreventAnOccupation}
and~\ref{thm:TwoParticlesTwoGoalsImpliesOneParticleOneGoal}, to show that it is
impossible to meet condition 2. 

Assume condition 1. Particles $a$ and $p$ start at $\mathbf{s}_{a}$ and
$\mathbf{s}_p$  respectively and at the conclusion of $\mathbf{m}$, the
locations  $\mathbf{g}_{a1}$ and $\mathbf{g}_{a2}$ are occupied. Now remove
particle $a$. According to
Lemma~\ref{thm:TwoParticlesTwoGoalsImpliesOneParticleOneGoal}, either
$\mathbf{g}_{a_1}$ or $\mathbf{g}_{a_2}$ must be occupied at the conclusion of
$\mathbf{m}$. Suppose without loss of generality that $\mathbf{g}_{a_1}$ is
filled.  By Lemma~\ref{thm:AdditionalParticlesCannotPreventAnOccupation},
adding an additional particle at location $\mathbf{s}_{\overline{a}}$ cannot
prevent $\mathbf{g}_{a_1}$ from being filled. However, to meet condition 2,
$\mathbf{g}_{a_1}$ must be vacant, thus no such gate is possible.  
\end{proof}

\section{Device and Gate Design}\label{sec:Design}
Now we consider actually designing clock sequence, logic gates, and wiring, making use of 2$\times$1 particles.

\paragraph{Choosing a clock sequence}

The \emph{clock sequence} is the ordered set of moves that are simultaneously applied to every particle in our workspace. We call this the clock sequence because, as in digital computers, this sequence is universally applied and keeps all logic synchronized.

A clock sequence determines the basic functionality of each gate.  To simplify implementation in the spirit of Reduced Instruction Set Computing (RISC), which uses a simplified set of instructions that run at the same rate, we want to use the same clock cycle for each gate and for \emph{all} wiring. 
Our early work in \cite{Becker2014} used a standard sequence  $\langle d,\ell,d,r \rangle$.  This sequence can be used to make {\sc and, or,} and {\sc xor} gates, and any of their inverses.  This sequence can also be used for \emph{wiring} to connect arbitrary inputs and outputs, as long as the outputs are below the inputs.  Unfortunately, $\langle d,\ell,d,r \rangle$ cannot move any particles upwards. To connect outputs as inputs to higher-level logic requires an additional reset sequence that contains a $\langle u \rangle$ command.  Therefore, including all four directions is a necessary condition for a valid clock sequence for computation that reuses gates.  The shortest sequence has four commands, each appearing once. We choose the sequence $\langle d,\ell,u,r \rangle$ and by designing examples, prove that this sequence is sufficient for logic gates, {\sc fan-out} gates, and wiring.

This clock sequence has the attractive property of being a clockwise (CW) rotation through the possible input sequences.  One could imagine our particle logic circuit mounted on a wheel rotating about an axis parallel to the ground. If the particles were moved by the pull of gravity, each counter-clockwise revolution would advance the circuit by one clock cycle. A gravity-fed hardware implementation of particle computation is shown in Fig.~\ref{fig:prototype}.

   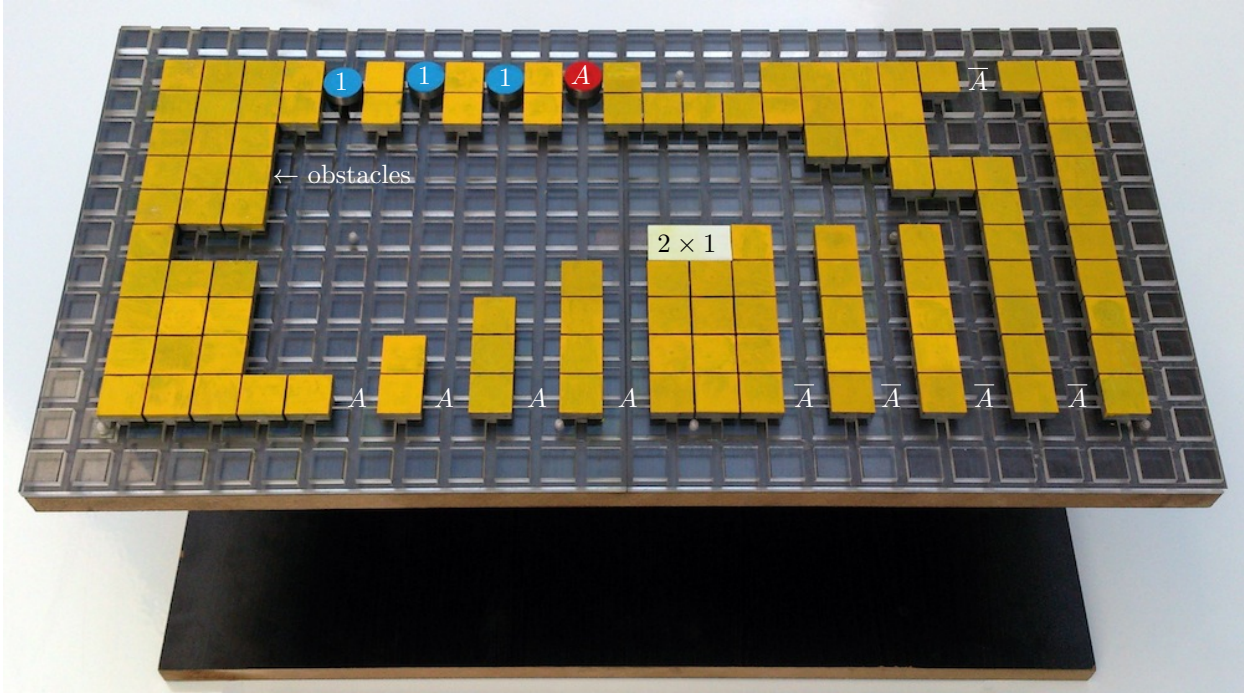
\begin{figure}
   \centering
   \href{http://youtu.be/EJSv8ny31r8}{
\begin{overpic}[width =\columnwidth]{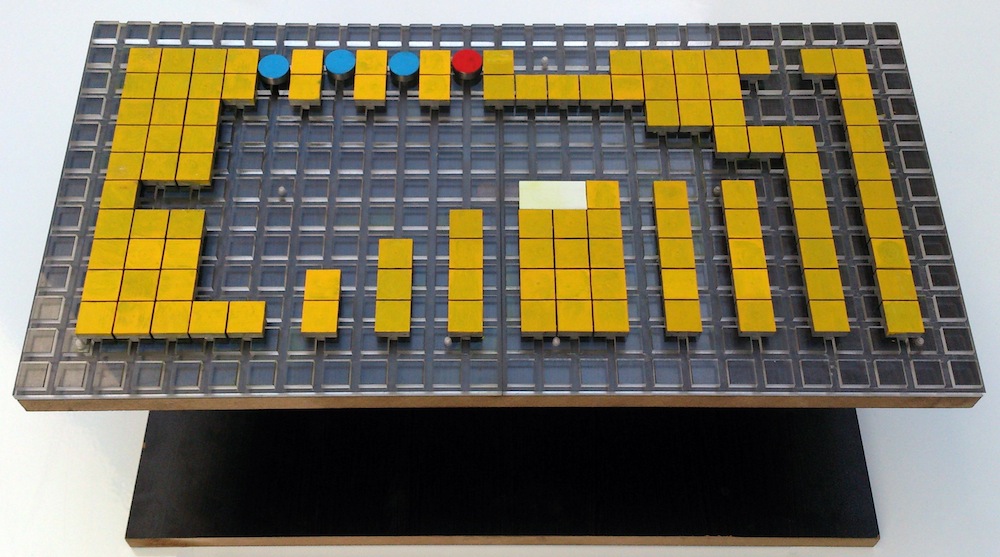}
\put(26,48.5){ \textcolor{white}{$1$}}
\put(32.7,49){ \textcolor{white}{$1$}}
\put(39.2,48.8){ \textcolor{white}{$1$}}
\put(45,49){ \textcolor{white}{$A$}}
\put(77,48.5){ \textcolor{white}{$\overline{A}$}}
\put(21,41){ \textcolor{white}{$\leftarrow$ obstacles}}
\put(52,35.5){ \textcolor{black}{$2 \times 1$}}
\put(27,23){ \textcolor{white}{$A$}}
\put(34,23){ \textcolor{white}{$A$}}
\put(41.5,23){ \textcolor{white}{$A$}}
\put(48.8,23){ \textcolor{white}{$A$}}
\put(63,23){ \textcolor{white}{$\overline{A}$}}
\put(70,23){ \textcolor{white}{$\overline{A}$}}
\put(77.5,23){ \textcolor{white}{$\overline{A}$}}
\put(85,23){ \textcolor{white}{$\overline{A}$}}
\end{overpic}}
\caption{
\label{fig:prototype}
Gravity-fed hardware implementation of  particle computation.  The reconfigurable prototype  is arranged using obstacle blocks (yellow) as a {\sc fan-out} gate using a 2$\times$1 particle (white), three supply particles (blue), and one red dual-rail input (red). This paper proves that such a gate is impossible using only 1$\times$1 particles. \href{http://youtu.be/H6o9DTIfkn0}{See the demonstrations in the video \cite{bmd+-pcdfbm-15}, \url{https://youtu.be/H6o9DTIfkn0}.} } 
\vspace{-1em}
\end{figure}

\emph{Limitations:} The clock sequence imposes constraints on the set of reachable positions after one cycle, as illustrated in Fig.~\ref{fig:Unreachable1cycle}.  
If at the completion of a $\langle d,\ell,u,r \rangle$ cycle a particle is located at  $(s_x,s_y)$, the potential locations at the end of the next cycle are any locations except $([s_x+1,\infty],s_y)$, and  $(s_x-1,[-\infty,s_y-1])$. 
For the particle to start at $(s_x,s_y)$ after a $r$ move, it must have been stopped by an obstacle at $(s_x+1,s_y)$, so $(s_x+1,s_y)$ is unreachable.  
Moving to the right of this obstacle requires a move of length $\lambda \ne 0$ in the down direction, followed by a $r$ move, followed by a move of $\lambda$ in the up direction.
However, $r$ is not the second move in the sequence. Because $r$ is the final move in the clock sequence, locations directly to the right of the start location are unreachable in one cycle.
Similarly, to end at any location $(s_x-1,g_y)$ with $g_y\le s_y$ requires both an obstacle at $(s_x,g_y)$ and an initial down move to at least $(s_x,g_y-1)$, but these requirements are contradictory.

 \begin{figure}
\centering
\begin{overpic}[width =.6\columnwidth]{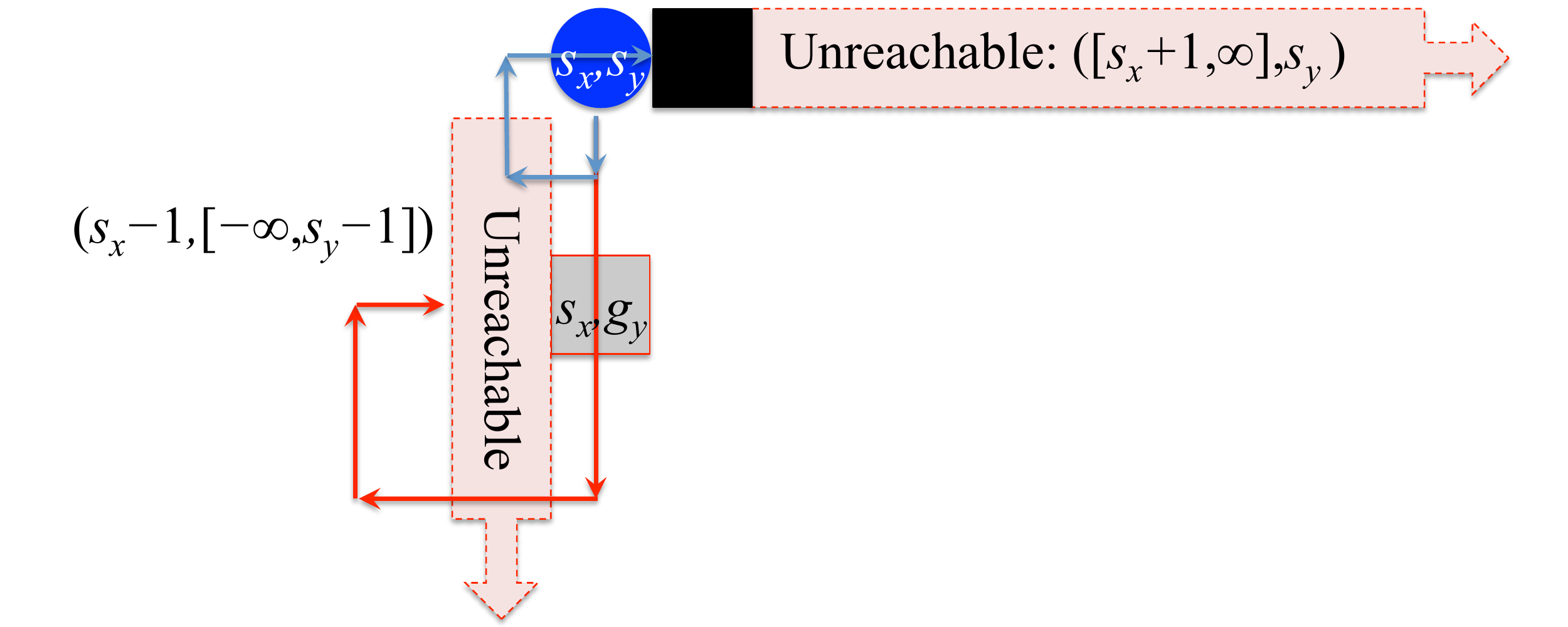}\end{overpic}
\caption{
\label{fig:Unreachable1cycle}
Two regions are unreachable in one $\langle d,\ell,u,r \rangle$  cycle.
}
\end{figure}

\paragraph{A {\sc fan-out} Gate}\label{sec:FanOut}
A {\sc fan-out} gate with two outputs implements the truth table in Table~\ref{tab:Fanout}.  
This cannot be implemented with 1$\times$1 particles and obstacles by Corollary \ref{cor:No1x1FanOut}.   Our technique uses 2$\times$1 particles.   A single-input, two-output {\sc fan-out} gate is shown in Fig.~\ref{fig:Fanout}.  This gate requires a dual-rail input, a supply particle, and a $2\times 1$ slider.  The  \emph{clockwise} control sequence $\langle d,\ell,u,r \rangle$  duplicates the dual-rail input.

The {\sc fan-out} gate can drive multiple outputs. In Fig.~\ref{fig:Fanout4} a single input drives four outputs.  This gate requires a dual-rail input, three supply particles, and a $2\times 1$ slider.  The \emph{clockwise} control sequence $\langle d,\ell,u,r \rangle$ quadruples the dual-rail input.
In general, an $n$-output {\sc fan-out} gate with control sequence $\langle d,\ell,u,r \rangle$ requires a dual-rail input, $n-1$ supply particles, and one $2\times 1$ slider. It requires an area of size $4n+7\times 2n+4$.
 \begin{figure}
\centering
 \vspace{1em}
\begin{overpic}[width =.5\columnwidth]{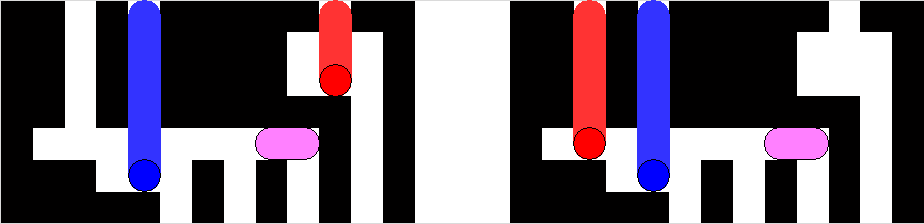}
\put(6.75,25){$A$} \put(15,25){$1$} \put(35,25){$\overline{A}$}
\put(62.5,25){$A$} \put(70,25){$1$} \put(90,25){$\overline{A}$}

\put(17.5,0){$A$}\put(24,0){$A$} \put(31.2,0){$\overline{A}$} \put(38,0){$\overline{A}$} 
\put(72.5,0){$A$} \put(79.5,0){$A$} \put(86.2,0){$\overline{A}$} \put(93,0){$\overline{A}$} 
\put(0,-5){ $\langle d \rangle$}\put(20,-5){$A=0$ }\put(70,-5){ $A=1$ }\end{overpic}
\vspace{1.5em}\\

\begin{overpic}[width =.5\columnwidth]{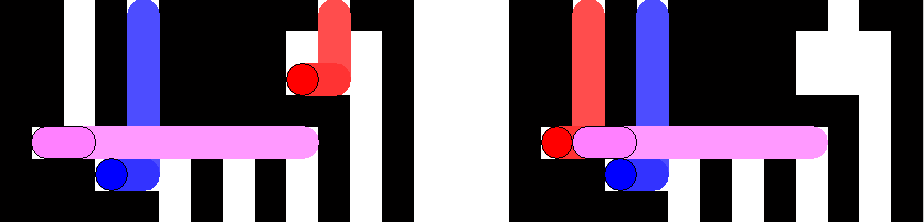}\put(0,-5){ $\langle d,\ell \rangle$}\put(20,-5){$A=0$ }\put(70,-5){ $A=1$ }\end{overpic}
\vspace{1.5em}\\

\begin{overpic}[width =.5\columnwidth]{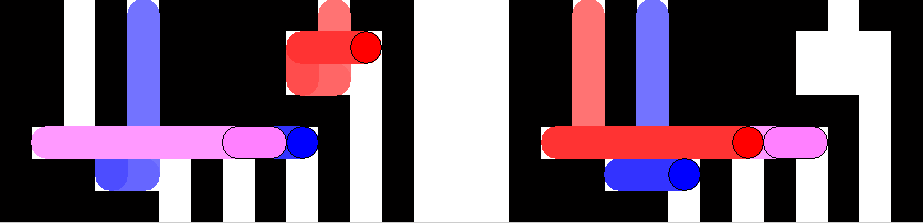}\put(0,-5){ $\langle d,\ell,u,r \rangle$}\put(20,-5){$A=0$ }\put(70,-5){ $A=1$ }
\put(17.5,0){$A$}\put(24,0){$A$} \put(31.2,0){$\overline{A}$} \put(38,0){$\overline{A}$} 
\put(72.5,0){$A$} \put(79.5,0){$A$} \put(86.2,0){$\overline{A}$} \put(93,0){$\overline{A}$} 
\end{overpic}
\vspace{1em}
\caption{
\label{fig:Fanout}
A single input, two-output {\sc fan-out} gate.  This gate requires a dual-rail input, a supply particle, and a $2\times 1$ slider.  The  \emph{clockwise} control sequence $\langle d,\ell,u,r \rangle$  duplicates the dual-rail input.
}
\end{figure}

 \begin{figure} 
\centering
\begin{overpic}[width =.5\columnwidth]{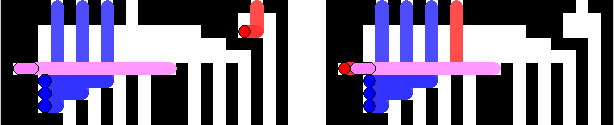}
\put(-2,-7){ $\langle d, \ell \rangle$}
\put(20,-7){$A=0$ }\put(70,-7){ $A=1$ }
\scriptsize
\put(9,21){$1$~~$1$~~$1$~~$A$} \put(41,21){$\overline{A}$} 
\put(61.2,21){$1$~~$1$~~$1$~~$A$} \put(93.5,21){$\overline{A}$} 

\put(10,-3){$A$~\,$A$~\,$A$~$A$} \put(30.5,-3){$\overline{A}$~\,$\overline{A}$~$\overline{A}$~\,$\overline{A}$} 
\put(63,-3){$A$~\,$A$~\,$A$~$A$} \put(83.6,-3){$\overline{A}$~\,$\overline{A}$~$\overline{A}$~\,$\overline{A}$} 
\end{overpic}
\vspace{3em}\\

\begin{overpic}[width =.5\columnwidth]{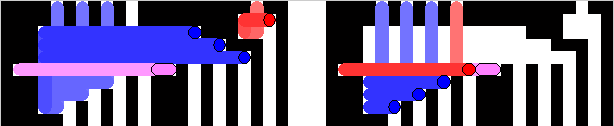}
\put(-2,-7){ $\langle d,\ell,u,r \rangle$}
\put(20,-7){$A=0$ }\put(70,-7){ $A=1$ }
\scriptsize
\put(9,21){$1$~~$1$~~$1$~~$A$} \put(41,21){$\overline{A}$} 
\put(61.2,21){$1$~~$1$~~$1$~~$A$} \put(93.5,21){$\overline{A}$} 

\put(10,-3){$A$~\,$A$~\,$A$~$A$} \put(30.5,-3){$\overline{A}$~\,$\overline{A}$~$\overline{A}$~\,$\overline{A}$} 
\put(63,-3){$A$~\,$A$~\,$A$~$A$} \put(83.6,-3){$\overline{A}$~\,$\overline{A}$~$\overline{A}$~\,$\overline{A}$} 
\end{overpic}
\vspace{1em}
\caption{\label{fig:Fanout4}
The {\sc fan-out} gate can drive multiple outputs. Here a single input drives four outputs.  This gate requires a dual-rail input, three supply particles, and a $2\times 1$ slider.  The \emph{clockwise} control sequence $\langle d,\ell,u,r \rangle$ quadruples the dual-rail input.
}
\end{figure}
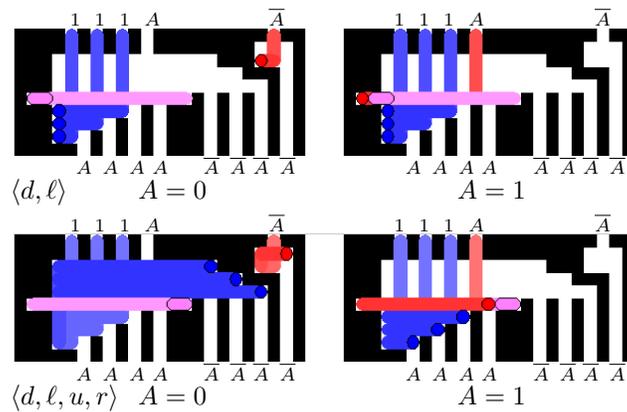

\paragraph{Data Storage\label{subsec:Storage}}

A general-purpose computer must be able to store data. 
 A $2\times1$ particle enables us to construct a read/writable data storage for one bit. A single-bit data storage latch is shown in Fig.~\ref{fig:Memory}.
This gate is conservative, and the memory state is given by the position of the $2\times 1$ slider: If the slider is low the memory state is true, if the slider is high the memory state is false. 
This gate implements the truth table in Table \ref{tab:memoryTruthTable}, and  has three inputs \emph{Set}, \emph{Clear}, or \emph{Read}. 
 Only one input should be true, making this a \emph{tri-rail} input. This input can be generated by logic on two dual-rail inputs: a Set/Clear input $S$ and a Read/Write input $R$, where  \emph{Set} $= S\wedge \overline{R}$,  \emph{Clear} $= \overline{S}\wedge \overline{R}$ , and  \emph{Read} $= R$.
Depending on which input is active, the \emph{clockwise} control sequence $\langle d,\ell,u,r \rangle$ will read, set, or clear the memory. 
The gate has a single output $M$ that reports the memory state after the inputs have been computed. The entire gadget requires a $16\times 8$ area.
  By combining an $n$-out {\sc fan-out} gate shown in Fig.~\ref{fig:Fanout4} with $n$ data storage devices, we can read from an $n$-bit memory. 

 \begin{figure*}
\begin{overpic}[width =\columnwidth]{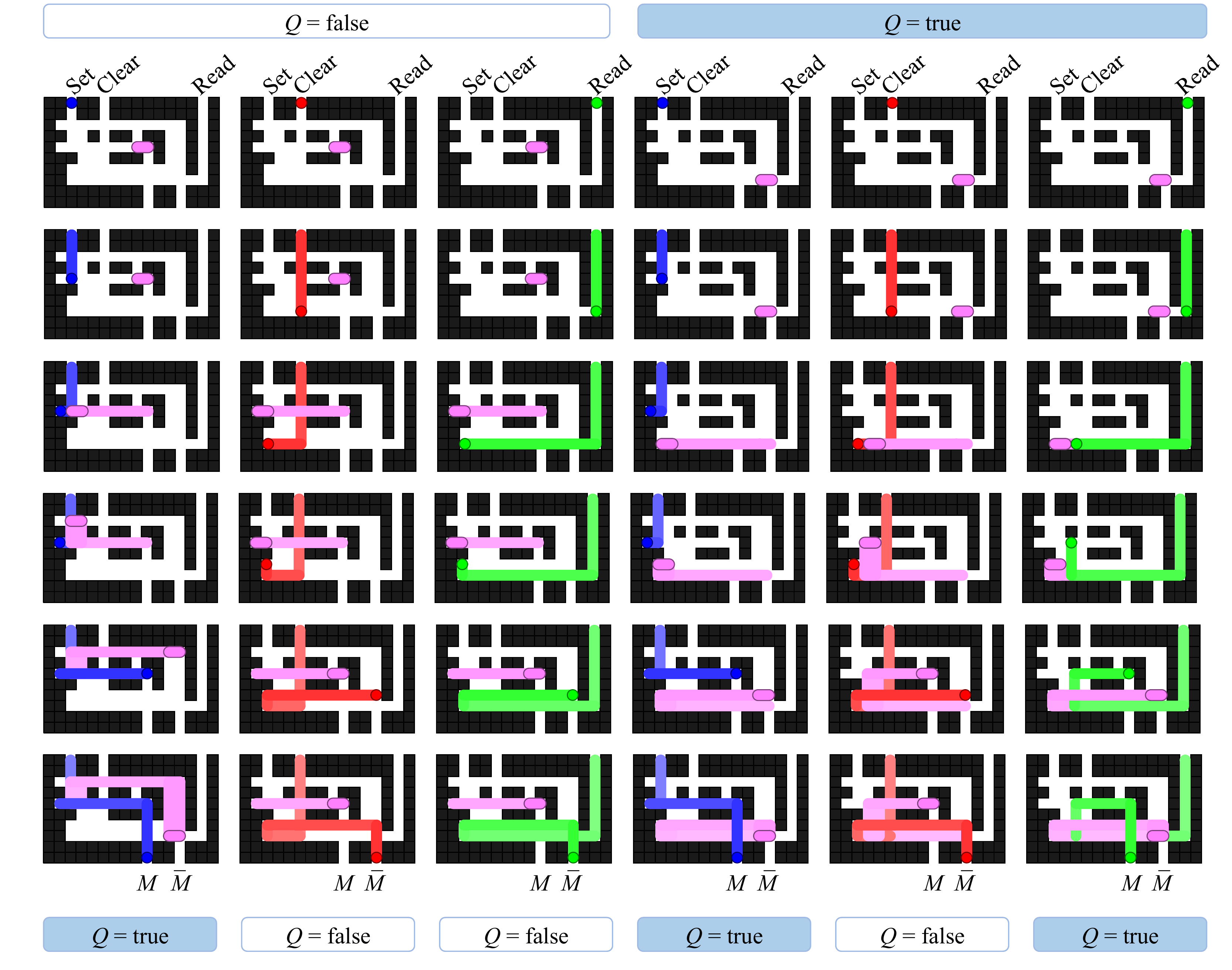}
\put(0,65.5){ \rotatebox{90}{$\langle \rangle$}}
\put(0,53.5){ \rotatebox{90}{$\langle d \rangle$}}
\put(0,43){ \rotatebox{90}{$\langle d,\ell \rangle$}}
\put(0,31){ \rotatebox{90}{$\langle d,\ell,u \rangle$}}
\put(0,19.5){ \rotatebox{90}{$\langle d,\ell,u,r \rangle$}}
\put(0,7.5){ \rotatebox{90}{$\langle d,\ell,u,r,d \rangle$}}
\end{overpic}
\caption{\label{fig:Memory}
A conservative, flip-flop memory gadget. 
 This gadget has a tri-rail input of \emph{Set}, \emph{Clear}, or \emph{Read}; and a $2\times 1$  state variable.  
The memory state $Q$ is given by the position of the $2\times 1$ slider: If the slider is low the memory state is true, if the slider is high the memory state is false.
Depending on which input is active, the \emph{clockwise} control sequence $\langle d,\ell,u,r \rangle$ will read, set, or clear the memory. 
The gate has a single output $M$ that reports the memory state $Q$ after the inputs have been computed. The entire gadget requires a $16\times 8$ area.
}
\end{figure*}

\begin{table}
\begin{displaymath}
\begin{array}{cccc|ccc}
\toprule
Q & S& C& R &Q & M & \overline{M}\\
\midrule     %
0 & 1 & 0 & 0 & 1 & 1 & 0 \\
0 & 0 & 1 & 0 & 0 & 0 & 1 \\
0 & 0 & 0 & 1 & 0 & 0 & 1 \\
1 & 1 & 0 & 0 & 1 & 1 & 0\\
1 & 0 & 1 & 0 & 0 & 0 & 1\\
1 & 0 & 0 & 1 & 1 & 1 & 0\\
\bottomrule
\end{array}
\end{displaymath}
  \caption{A single-bit data storage latch with state $Q$, inputs \emph{Set}, \emph{Clear}, or \emph{Read}, and outputs representing the memory state $M$ and the inverse $\overline{M}$.  \label{tab:memoryTruthTable}}
\end{table}
  
\paragraph{A Binary Counter}\label{sec:binaryCounter}
  Using the {\sc fan-out} gate from Section \ref{sec:FanOut} we can generate arbitrary Boolean logic.  The half adder from Fig.~\ref{fig:HalfAdder} requires a single {\sc fan-out} gate, one {\sc and}, and one {\sc xor} gate.
  
We illustrate how several gates can be combined by constructing a binary counter, shown in Fig.~\ref{fig:Counter}. Six logic gates are used to implement a 3-bit counter. A block diagram of the device is shown in Fig.~\ref{fig:CounterBlockDiagram}. The counter  requires three {\sc fan-out} gates, two summers ({\sc xor}) gates, and one carry ({\sc and}) gate. Six $1\times1$ particles and three $2\times1$ particles are used.  The counter has three levels of gates actuated by CW sequence $\langle d,\ell,u,r \rangle$ and requires three interconnection moves $\langle d,\ell,u,r \rangle$, for a total of 24 moves (6 cycles) per count. 

\begin{figure}\centering
 \begin{overpic}[width =.6\columnwidth]{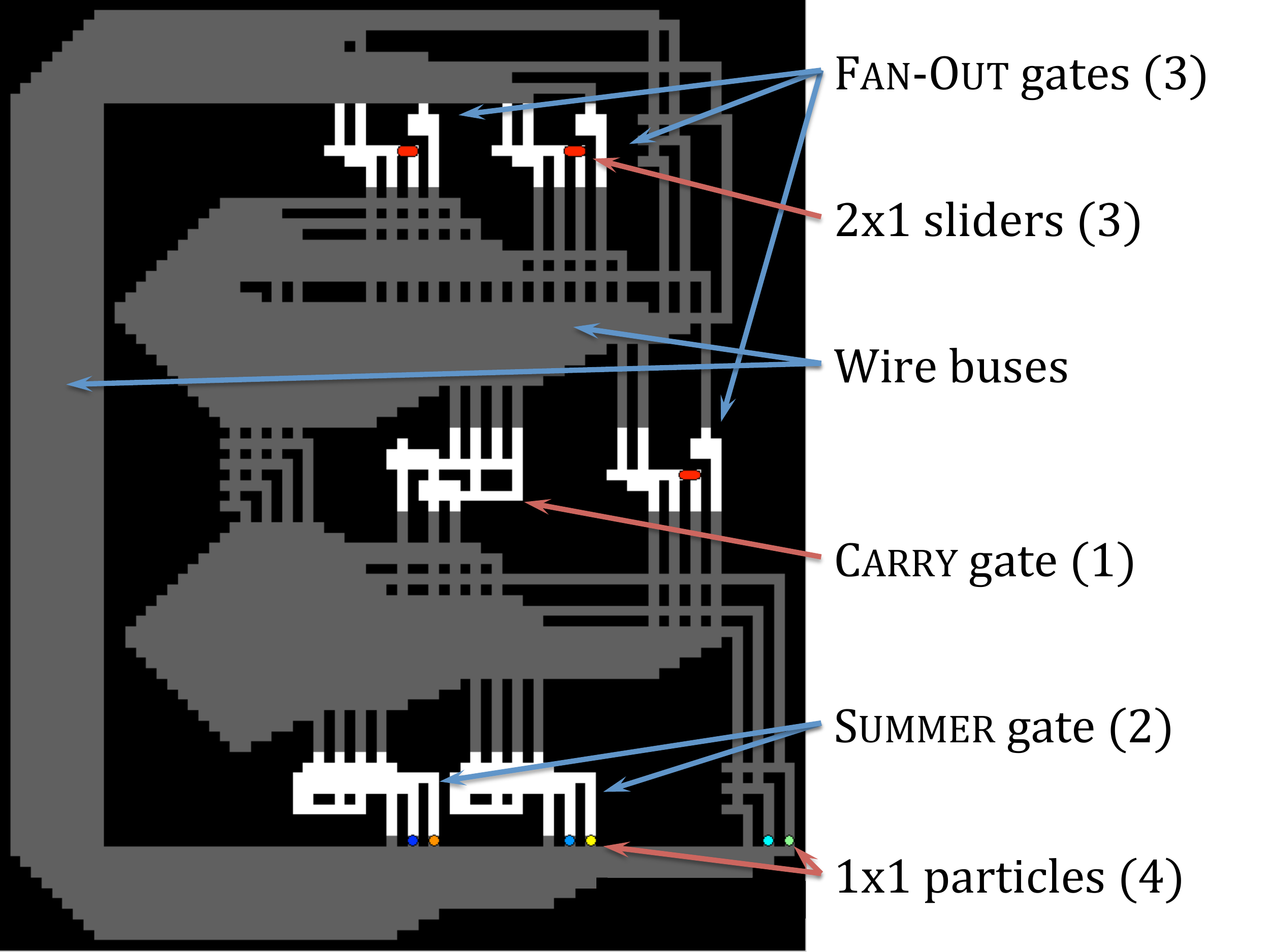}
 \put(30.5,6){\scalebox{0.5}{\textcolor{white}{$b_2  \bar{b}_2\,1$}}}
 \put(42.8,6){\scalebox{0.5}{\textcolor{white}{$b_1  \bar{b}_1\,1$}}}
  \put(58.5,6){\scalebox{0.5}{\textcolor{white}{$b_0  \bar{b}_0\,1$}}}
   \end{overpic}
\vspace{-1em}\\
\caption{
\label{fig:Counter}
A three-bit counter implemented with particles. The counter requires three {\sc fan-out} gates, two summer gates, and one carry gate.  Six 1$\times$1 particles and three 2$\times$1 particles are used.  The gates and wire buses are actuated by the CW sequence $\langle d,\ell,u,r \rangle$. See video at \url{https://youtu.be/QRAOaLZjuBY?t=4m9s} for animation.
}
\end{figure}

 \begin{figure}
\centering
 \begin{overpic}[width =.5\columnwidth]{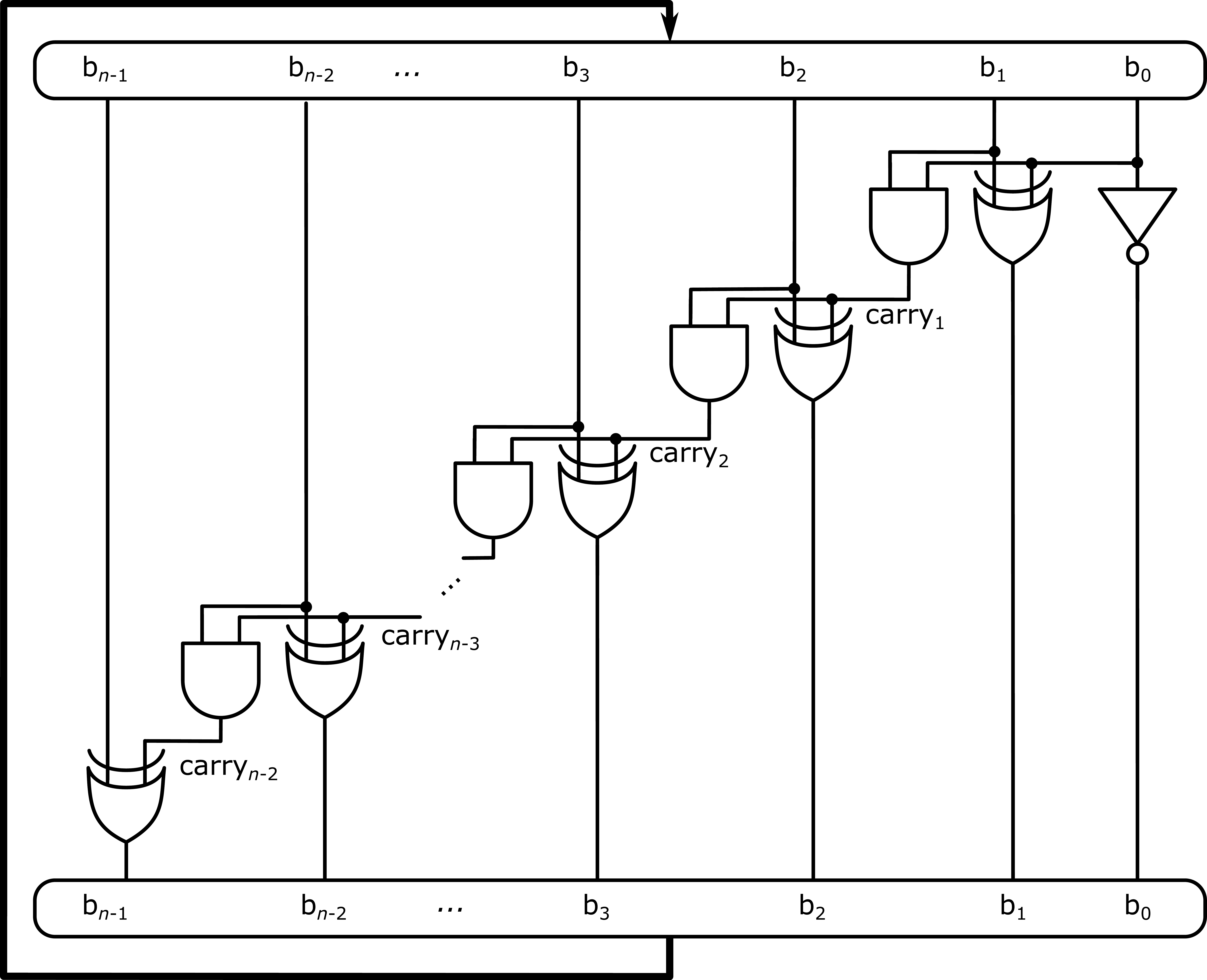}\end{overpic}
\caption{
\label{fig:CounterBlockDiagram}
Gate-level diagram for an $n$-bit counter.  
}
\vspace{-1em}
\end{figure}


\paragraph{Scaling Issues}\setcounter{paragraph}{0}
 Particle computation requires multiple clock cycles, workspace area for gates and interconnections, and many particles.  In this section we analyze how these scale with the size of the counter, using Fig.~\ref{fig:CounterBlockDiagram} as a reference.   

\begin{description}
\item[Gates.]  An $n$-bit counter requires $3(n-1)$ gates: $n$ {\sc fan-out} gates, $n-1$ summers ({\sc xor}) gates, and $n-2$ carry ({\sc and}) gates. 
\item[Particles.] We require $n$ 1$\times$1 particles, one for each bit, and $n$ 2$\times$1 particles, one for each {\sc fan-out} gate.
\item[Propagation delay.] The counter requires $n$ stages of logic and $n$ corresponding wiring stages.  Each stage requires a complete clock cycle $\langle d,\ell,u,r \rangle$ for a total of 8$n$ moves.
\end{description}
These requirements are comparable to a ripple-carry adder:  the delay for $n$ bits is $n$ delays  and requires $5(n-1)+2$ gates.
Numerous other schemes exist to speed up the computation; however, using discrete gates allows us to use standard methods for translating a Boolean expression into gate-level logic.  If speed were critical, instead of using discrete gates, we could engineer the workspace to compute the desired logic directly.  

\paragraph{Optimal Wiring Schemes}\label{sec:wiring}
With our current CW clock cycle, we cannot have outputs from the same column as
inputs---outputs must be either one to the right or three to the left.  
Choosing one of these results in horizontal shifts at each stage and thus
requires spreading out the logic gates. A more compact wiring scheme cycles
through three layers that each shift right by one, followed by one layer that
shifts left by three.  We also want the wiring to be tight left-to-right.  If our
height is also limited, \emph{wire buses}, shown in Fig.~\ref{fig:Counter} provide a compact solution. 

\paragraph{Optimized logic}\label{sec:optimizedLogic}
The particle-computation presented in this section is general purpose, and Fig.~\ref{fig:Counter} illustrates how a set of gate components can be composed to compute arbitrary logic.
Single-purpose logic can often be more compact, as shown by Fig.~\ref{fig:compactCounter}, which shows three counters that use less area and fewer particles than Fig.~\ref{fig:Counter}.

 \begin{figure}
\centering
 \begin{overpic}[width =.75\columnwidth]{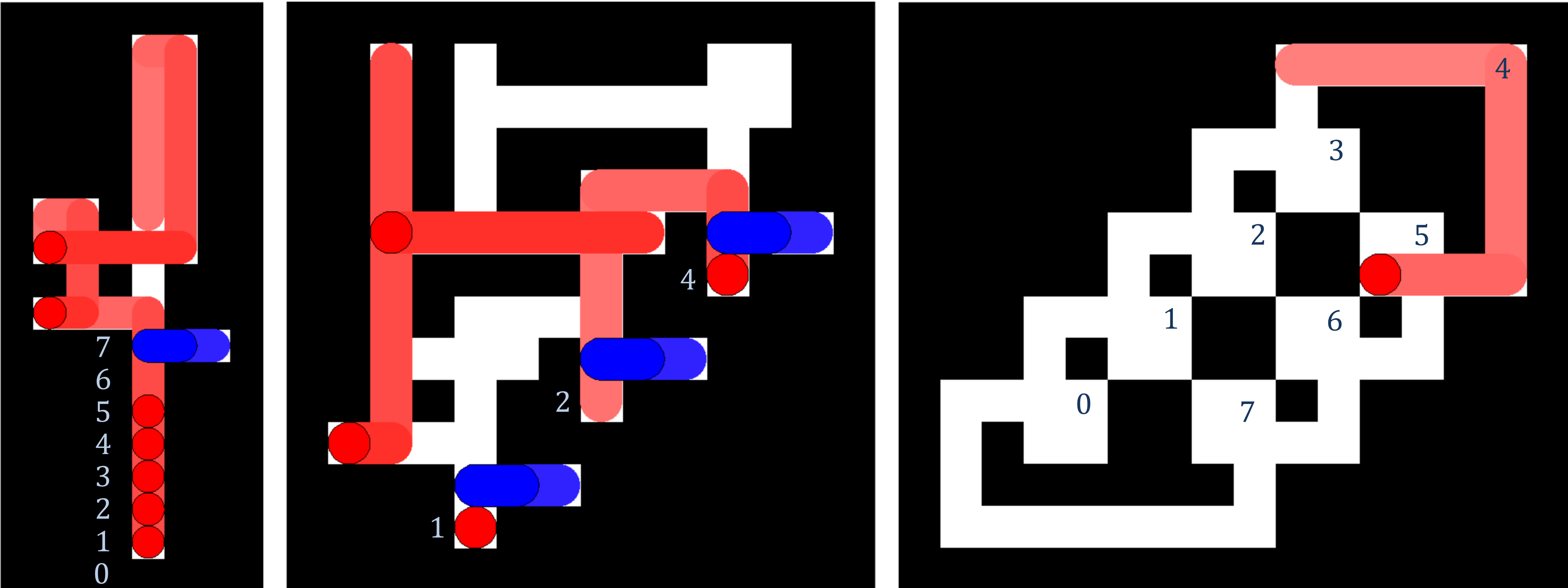}\end{overpic}
\caption{
\label{fig:compactCounter}
Custom logic can be compact. Each counter above uses the clock sequence $\langle d,\ell,u,r \rangle$ and resets after 8 cycles (32 moves). (Left) an arrangement that cyclically counts from zero to seven: seven particles, one 2$\times$1 slider, 8$\times$18 area. (Center) a binary counter with three bits: three particles, three 2$\times$1 sliders, 14$\times$14 area. (Right) a gadget that resets every 32 moves: one particle, 16$\times$14 area.
}
\vspace{-1em}
\end{figure}

\section{Conclusion}\label{sec:conclusion}

This paper analyzes the problem of steering many particles with uniform inputs
in a 2D environment containing obstacles. We design environments that
can efficiently perform matrix operations on groups of particles in parallel,
including a matrix permutation requiring only four moves for any number of
particles. These matrix operations enable us to prove that the general motion
planning problem is {\sc pspace}-complete.

We also introduce a new model for mechanical computation.  We
(1) prove the insufficiency of unit-size particles for gate fan-out; 
(2) establish the necessity of dual-rail logic for Boolean logic;  
(3) design {\sc fan-out} gates and memory latches by employing slightly different particles; and 
(4) present an architecture for device integration, a common clock sequence, and a binary counter.
 
  There remain many interesting problems to solve. We are motivated by
practical challenges in steering micro-particles through vascular networks,  which
are common in biology. Though some are two-dimensional, including the leaf
example in Fig.~\ref{fig:vascularNetwork} and endothelial networks on the
surface of organs, many of these networks are three dimensional.
Magnetically actuated systems are capable of providing 3D control inputs, but
control design poses additional challenges.

   We investigate a subset of control in which all particles move
maximally. Future work should investigate more general motion---what happens
if we can move all the particles a discrete distance or
along an arbitrary curve? We also abstracted important practical constraints, e.g.,
ferromagnetic objects tend to clump in a magnetic field, and most magnetic fields are not perfectly uniform.

Finally, our research has potential applications in micro-construction, microfluidics, and
nano-assembly.  These applications require additional theoretical analysis to
model heterogeneous objects and objects that bond when forced together, e.g.,
MEMS components and molecular chains.


\section*{Acknowledgments}

We thank an anonymous reviewer for carefully going through our work and making numerous 
constructive suggestions that helped to improve the presentation of our paper.
We thank Hamed Mohtasham Shad for building and testing the first
experimental tilt tables that brought these algorithms to life.  We acknowledge
the helpful discussion and motivating experimental efforts with \emph{T.
pyriformis} cells by Yan Ou and Agung Julius at RPI and Paul Kim and MinJun Kim
at Drexel University.  
Preliminary versions of 
        Section~\ref{sec:mazes} and Section~\ref{sec:matrices} are main topics of our paper~\cite{Becker2013f} 
        with an extra result proving the system to give rise to {\sc pspace}-completeness in Section~\ref{subsec:pspaceComplete} from paper  \cite{Becker2014}.
        The particle logic in Sections~\ref{sec:logic} and~\ref{sec:Design} was introduced in \cite{Becker2014} and completed in paper \cite{shad2015particle}.

This work has been partially supported by the National Science Foundation (Grant No.\ 
\href{http://nsf.gov/awardsearch/showAward?AWD_ID=1553063}{[IIS-1553063]}
 and 
\href{http://nsf.gov/awardsearch/showAward?AWD_ID=1619278}{[IIS-1619278]}).

\bibliographystyle{abbrv}
\bibliography{refs}

\end{document}